\documentclass[11pt,english,onecolumn, draftcls]{IEEEtran}
\usepackage[T1]{fontenc}
\usepackage[latin9]{inputenc}
\usepackage{color}
\usepackage{calc}
\usepackage{amsmath}
\usepackage{amssymb}
\usepackage{graphicx}
\usepackage{esint}

\makeatletter
  \newtheorem{assumption}{Assumption}
  \newtheorem{definitn}{Definition}
  \newtheorem{remrk}{Remark}
  \newtheorem{lemma}{Lemma}
  \newtheorem{example}{Example}
  \newtheorem{constraints}{Constraints}
  \newtheorem{thm}{Theorem}
  \newtheorem{cor}{Corollary}
  \newtheorem{problem}{Problem}

\usepackage{cite}

\author{Xiongbin~Rao,~\IEEEmembership{Student~Member,~IEEE} and~Vincent~K.~N.~Lau,~\IEEEmembership{Fellow,~IEEE}
\thanks{The authors are with the Department of Electronic and Computer Engineering (ECE), the Hong Kong University of Science and Technology (HKUST), Hong Kong (e-mail: \{xrao,eeknlau\}@ust.hk).}%
 }

\makeatother

\usepackage{babel}
\begin{document}

\title{Interference Alignment with Partial CSI Feedback in MIMO Cellular
Networks}
\maketitle
\begin{abstract}
Interference alignment (IA) is a linear precoding strategy that can
achieve optimal capacity scaling at high SNR in interference networks.
However, most existing IA designs require full channel state information
(CSI) at the transmitters, which would lead to significant CSI signaling
overhead. There are two techniques, namely \emph{CSI quantization}
and \emph{CSI feedback filtering}, to reduce the CSI feedback overhead.
In this paper, we consider IA processing with CSI feedback filtering
in MIMO cellular networks. We introduce a novel metric, namely the
\emph{feedback dimension}, to quantify the first order CSI feedback
cost associated with the CSI feedback filtering. The CSI feedback
filtering poses several important challenges in IA processing. First,
there is a hidden \emph{partial CSI knowledge} constraint in IA precoder
design which cannot be handled using conventional IA design methodology.
Furthermore, existing results on the feasibility conditions of IA
cannot be applied due to the partial CSI knowledge. Finally, it is
very challenging to find out how much CSI feedback is actually needed
to support IA processing. We shall address the above challenges and
propose a new IA feasibility condition under partial CSIT knowledge
in MIMO cellular networks. Based on this, we consider the CSI feedback
profile design subject to the degrees of freedom requirements, and
we derive closed-form trade-off results between the CSI feedback cost
and IA performance in MIMO cellular networks. \end{abstract}

\section{Introduction}

It is well known that inter-cell interference is one of the most important
performance bottlenecks in wireless networks. There are many works
on interference mitigation techniques and conventional approaches
either treat interference as noise or rely on \emph{interference avoidance}
by means of channel orthogonalization \cite{sreekanth2008sum}. However,
these schemes are far from optimal \cite{host2005multiplexing}. Recently,
interference alignment (IA) was proposed as an effective means to
mitigate interference in $K$-user interference channels \cite{cadambe2008interference,gomadam2011distributed}.
By aligning the interference from different transmitters (Txs) into
a lower dimensional subspace at each receiver (Rx), IA can achieve
the optimal capacity scaling with respect to (w.r.t.) SNR. As such,
there is a surge in the research interest of IA and it has been extended
to other topologies such as MIMO cellular networks in \cite{ruan2012interference,zhuang2011interference}.

Despite the fact the IA can achieve substantial throughput gain, conventional
IA designs \cite{cadambe2008interference,gomadam2011distributed,ruan2012interference,zhuang2011interference}
require full channel state information at the Tx side (CSIT). Such
full CSIT requirement is quite difficult to achieve in practice due
to limited CSI feedback capacity in the reverse link in practice.
As such, naive IA design will be very sensitive to CSIT errors \cite{thukral2009interference,krishnamachari2009interference}
and it is important to take into account the CSI feedback constraint
in the IA design. There are, in general, two ways to reduce the CSI
feedback overhead, namely \emph{CSI quantization} and \emph{CSI filtering}.
In \cite{thukral2009interference,krishnamachari2009interference},
the authors considered using Grassmannian codebooks to quantize and
feedback the channel direction information (CDI) for IA processing.
In \cite{rao2012limited,el2011grassmannian}, some adaptive quantization
schemes are proposed to exploit the channel statistics so as to enhance
the limited CSI feedback efficiency. However, these schemes considered
CSI quantization of the full CDI in the interference networks only. 

In fact, full CDI may not always be needed to achieve IA processing
at the Txs. We illustrate two examples in which substantially reduced
CSI is fed back to achieve IA processing. Furthermore, the \emph{CSI
quantization} and the \emph{CSI filtering} techniques are complementary
to each other and in some situations, the CSI filtering will be a
first order contributor towards enhancing the CSI feedback efficiency
in MIMO cellular networks. The CSI filtering techniques to reduce
feedback overhead are relatively less explored. In \cite{de2012interference},
a CSI filtering scheme by CSI truncation is proposed to reduce the
CSI feedback in MIMO interference network. In \cite{suh2011downlink},
a CSI filtering scheme with zero-forcing IA is proposed to eliminate
the intercell CSI feedback in MIMO cellular networks. However, a more
systematic understanding is still needed to determine how much CSI
feedback is required for IA processing. In this paper, we propose
a systematic framework of CSI filtering and analyze the associated
tradeoff between CSI feedback cost and IA degrees of freedom (DoF)
performance in MIMO cellular networks. There are several unique challenges
that need to be tackled. 
\begin{itemize}
\item \textbf{How to quantify the CSI Feedback Cost?} It may be natural
to measure the CSI feedback cost in MIMO cellular networks in terms
of the total number of the feedback bits. However, this metric mixes
the \emph{CSI filtering} and \emph{CSI quantization} together. To
obtain some key design insights, it is desirable to have a metric
that can solely focus on the CSI filtering aspect because the CSI
quantization is complementary and can always be considered on top
of the CSI filtering as in Figure \ref{fig:Role-of-CSI}. 
\item \textbf{IA Feasibility Conditions under Partial CSI Feedback:} It
is well known that the IA scheme is not always feasible and the feasibility
conditions are topology specific. The IA feasibility condition is
studied for MIMO interference channels in \cite{yetis2010feasibility,gonzalez2012general,razaviyayn2011degrees,ruan2012feasibility},
and for MIMO cellular networks in \cite{liu2013feasibility}. However,
these works have assumed full CSIT%
\footnote{For instance, in conventional IA formulation \cite{yetis2010feasibility,gonzalez2012general,razaviyayn2011degrees,ruan2012feasibility},
the IA precoders / decorrelators $\{\mathbf{V}_{i},\mathbf{U}_{i}:\forall i\}$
are found to be a function of the \emph{entire} CSIs $\{\mathbf{H}_{ji}:\forall j,i\}$
such that $\textrm{rank}\left(\mathbf{U}_{j}^{\dagger}\mathbf{H}_{jj}\mathbf{V}_{j}\right)=d,\;\mathbf{U}_{j}^{\dagger}\mathbf{H}_{ji}\mathbf{V}_{i}=\mathbf{0},\forall j,i\neq j$. %
}and hence the precoders can be designed as a function of the full
CSI. While in MIMO cellular networks with CSI feedback filtering,
the precoders can only be designed based on the \emph{partial CSI
knowledge} from CSI feedback filtering and hence the IA feasibility
conditions are different.
\item \textbf{CSI Feedback Design:} Further, it remains a question what
is the CSI filtering scheme with the \emph{least} amount of CSI feedback
overhead to support the required IA DoFs for a given antenna configuration.
Such a question involves minimization of the CSI feedback cost subject
to IA feasibility constraint. However, this problem is highly non-trivial
because of the combinatorial nature of CSI filtering scheme design.
\end{itemize}

In this paper, we will address the challenges above as follows. We
first define a novel CSI feedback cost metric, namely the \emph{CSI
feedback dimension}. The CSI feedback dimension enables us to isolate
the CSI quantization effects from the CSI filtering design so as to
obtain tractable and first order design insights. Based on the proposed
metric, we propose the idea of IA processing under partial CSI feedback
in MIMO cellular networks. After that, we investigate the feasibility
conditions and derive the associated precoder / decorrelator solutions
for IA under a given partial CSI feedback scheme. Based on these results,
we attempt to find out the \emph{least} amount of CSI feedback overhead
by formulating the problem of minimizing CSI feedback dimension subject
to IA constraints with a given IA DoFs in the network for a given
antenna configuration. Using specific insights from the problem, we
derive a low complexity asymptotically optimal solution and obtain
closed-form tradeoff results between the number of DoFs and the CSI
feedback dimension. Finally, we compare the proposed IA design with
various state-of-the-art baselines and illustrate that the proposed
solution achieves significant CSI feedback cost reduction in MIMO
cellular networks. 

\textit{Notation}s: Uppercase and lowercase boldface letters denote
matrices and vectors respectively. The operators $(\cdot)^{T}$, $(\cdot)^{\dagger}$,
$\textrm{rank}(\cdot)$, $|\cdot|$, $\textrm{tr}(\cdot)$, $\textrm{dim}_{s}(\cdot)$,
$\otimes$, $\left\lfloor \cdot\right\rfloor $, $\left\lceil \cdot\right\rceil $,
$\left\Vert \cdot\right\Vert $ and $\textrm{vec}(\cdot)$ are the
transpose, conjugate transpose, rank, cardinality, trace, dimension
of subspace, Kronecker product, integer floor, integer ceiling, Frobenius
norm and vectorization respectively; $\mathbf{I}_{d}$, $\mathbb{Z}$
and $\mathbb{U}(A,B)=\left\{ \mathbf{U}\in\mathbb{C}^{A\times B}:\mathbf{U}^{\dagger}\mathbf{U}=\mathbf{I}\right\} $
denote the identity matrix, the set of non-negative integers, and
the set of $A\times B$ ($A\geq B$) semi-unitary matrices respectively;
$\mathbb{P}(\mathbf{A})=\{a\mathbf{A}:a\in\mathbb{C}\}$ and $\textrm{span}(\{\mathbf{A}_{i}\})$
denotes the vector space spanned by all the column vectors of the
matrices in $\{\mathbf{A}_{i}\}$, and $d\mid M$ denotes that integer
$M$ is divisible by integer $d$.

\section{System Model}

\subsection{MIMO Cellular Networks}

Consider a MIMO cellular network with $G$ base stations (BSs) and
each BS serves $K$ mobile stations (MSs) as illustrated in Figure
2. Consider that each BS and MS are equipped with $N$ and $M$ antennas
respectively, and $d$ data streams are transmitted to each MS from
its serving BS. We focus on the case when $M\leq(G-1)Kd+d$ because
otherwise, i.e., $M>(G-1)Kd+d$, the number of antennas at the MS
is over-sufficient to cancel all the inter-cell interference using
pure zero forcing at the MS.

Denote the transmit SNR at each BS as $P$, the $k$-th MS of BS $j$
as the $(j,k)$-th MS, the channel matrix from the $i$-th BS to the
$(j,k)$-th MS as $\mathbf{H}_{jk,i}\in\mathbb{C}^{M\times N}$. The
received signal at the $(j,k)$-th MS is given by:

\[
\mathbf{y}_{jk}=\mathbf{U}_{jk}^{\dagger}\left(\mathbf{H}_{jk,j}\mathbf{V}_{jk}\mathbf{x}_{jk}+\vphantom{\sum_{p=1}^{K}}\right.\underset{\textrm{intra-cell interference}}{\underbrace{\sum_{\underset{\neq k}{p=1}}^{K}\mathbf{H}_{jk,j}\mathbf{V}_{jp}\mathbf{x}_{jp}}}+\underset{\textrm{inter-cell interference}}{\underbrace{\sum_{\underset{\neq j}{i=1}}^{G}\sum_{p=1}^{K}\mathbf{H}_{jk,i}\mathbf{V}_{ip}\mathbf{x}_{ip}}}\left.\vphantom{\sum_{p=1}^{K}}+\mathbf{n}_{jk}\right),\forall j,k
\]
where $\mathbf{x}_{jk}\sim\mathcal{CN}(\mathbf{0},\;\frac{P}{Kd}\mathbf{I}_{d})$
is the encoded information symbol sent from the $j$-th BS to the
$(j,k)$-th MS, $\mathbf{V}_{jk}\in\mathbb{C}^{N\times d}$ and $\mathbf{U}_{jk}\in\mathbb{C}^{M\times d}$
are the corresponding precoder and decorrelator matrix for the $(j,k)$-th
MS, $\mathbf{n}_{jk}\sim\mathcal{CN}(\mathbf{0},\;\mathbf{I}_{M})$
is the white Gaussian noise.
\begin{assumption}
[Channel Matrices]\label{Channel-MatricesAssume}Assume the elements
of $\mathbf{H}_{jk,i}$ are i.i.d. complex Gaussian random variables
with zero mean and unit variance. The CSIs are observable at the MSs
and the CSI feedback from the $(j,k)$-th MS will be received error-free
by BS $j$. Furthermore, we assume the BSs $\{1,\cdots,G\}$ have
backhaul connections such that the feedback CSI can be shared among
them.
\end{assumption}

\subsection{CSI Feedback Filtering and Feedback Cost}

The CSI feedback reduction in MIMO cellular networks contains two
processes in general, namely the CSI filtering and the CSI quantization
as illustrated in Figure \ref{fig:Role-of-CSI}. To simplify the analysis,
we shall consider these two factors separately. We consider CSI filtering
only in Sections II-IV (no quantization is performed) and then analyze
the effects of CSI quantization (block (b)) in Section V. Since IA
processing aims at nulling off interferences at the MS, only the CDI%
\footnote{For example, in IA designs, if $\mathbf{\mathbf{U}}^{\dagger}\mathbf{H}\mathbf{V}=\mathbf{0}$,
then we have $\mathbf{\mathbf{U}}^{\dagger}(a\mathbf{H})\mathbf{V}=0$,
$\forall a\in\mathbb{C}$. Hence, it is sufficient to feeding back
the CDI for IA, i.e., $\mathbb{P}(\mathbf{H})=\{a\mathbf{H}:a\in\mathbb{C}\}$,
which is contained in $\mathbb{G}(1,MN)$ \cite{dai2008quantization}.%
}, i.e., $\mathbb{P}(\mathbf{H}_{jk,i})=\{a\mathbf{H}_{jk,i}:a\in\mathbb{C}\}$,
$\forall j,k,i$, is required to design the IA transceivers. Hence,
we shall consider CSI feedback over the Grassmannian manifold. Let
$\mathcal{H}_{jk}=(\mathbf{H}_{jk,1},\mathbf{H}_{jk,2},\cdots\mathbf{H}_{jk,G})\in\prod_{i=1}^{G}\mathbb{C}^{M\times N}$
be the tuple of\emph{ }CSI matrices observed at the $(j,k)$-th MS
and let $\mathbb{G}(A,B)$ be the Grassmannian manifold of $A$ dimensional
subspaces in $\mathbb{C}^{B\times1}$. The \emph{CSI feedback filtering}
at each MS is modeled by the following model.

\begin{figure}
\begin{centering}
\includegraphics{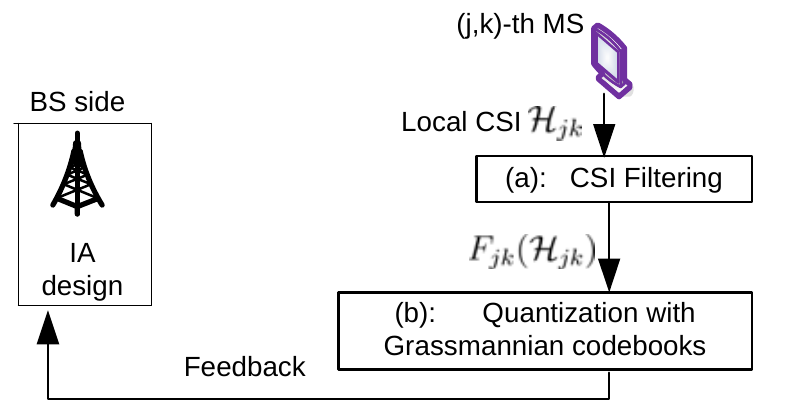}
\par\end{centering}

\caption{\label{fig:Role-of-CSI}Role of CSI filtering in the CSI feedback
reduction.}
\end{figure}

\begin{definitn}
[CSI Feedback Filtering]The partial CSI feedback generated by the
$(j,k)$-th MS is a $l_{jk}$-tuple, which  can be characterized by
a \emph{feedback filtering function} $F_{jk}$: $\prod_{i=1}^{G}\mathbb{C}^{M\times N}\rightarrow\prod_{i=1}^{l_{jk}}\mathbb{G}(A_{jk}^{[i]},B_{jk}^{[i]})$.
That is:
\begin{equation}
\mathbb{\mathcal{H}}_{jk}^{fed}=F_{jk}(\mathcal{H}_{jk})\label{eq:feedback-info}
\end{equation}
where $l_{jk}$ denotes the number of subspaces in $\mathbb{\mathcal{H}}_{jk}^{fed}$,
$\mathbb{\mathcal{H}}_{jk}^{fed}\in\mathbb{G}(A_{jk}^{[1]},B_{jk}^{[2]})\times\mathbb{G}(A_{jk}^{[2]},B_{jk}^{[2]})\times\cdots\mathbb{G}(A_{jk}^{[l_{jk}]},B_{jk}^{[l_{jk}]})$
is the partial CSI generated at the $(j,k)$-th MS, $\mathbb{G}(A_{jk}^{[i]},B_{jk}^{[i]})$
is the associated Grassmannian manifold containing the $i$-th element
in the CSI feedback tuple $\mathbb{\mathcal{H}}_{jk}^{fed}$, and
$A_{jk}^{[i]},$ $B_{jk}^{[i]}$ are parameters characterizing that
the $i$-th element in $\mathbb{\mathcal{H}}_{jk}^{fed}$ is a $A_{jk}^{[i]}$-dimensional
subspace in $\mathbb{C}^{B_{jk}^{[i]}\times1}$. \hfill \QED
\end{definitn}

In other words, the output of the CSI feedback filtering is a tuple
of subspaces where each subspace corresponds to a point in the associated
Grassmannian manifold \cite{dai2008quantization}. For example, consider
two CSI matrices $\mathbf{H}_{1},\mathbf{H}_{2}\in\mathbb{C}^{2\times3}$.
If we feedback $\mathbb{P}(\mathbf{H}_{1})=\{a\mathbf{H}_{1}:a\in\mathbb{C}\}$,
$\mathbb{P}(\mathbf{H}_{2})=\{a\mathbf{H}_{2}:a\in\mathbb{C}\}$,
then this corresponds to the feedback filtering function $F=\left(\begin{array}{cc}
\mathbb{P}(\mathbf{H}_{1}), & \mathbb{P}(\mathbf{H}_{2})\end{array}\right)\in\mathbb{G}(1,6)\times\mathbb{G}(1,6)$. Note that under \emph{given} feedback filtering functions $\{F_{jk}\}$,
the partial CSI $\{F_{jk}(\mathcal{H}_{jk})\}$ will be fed back to
the BSs for the IA precoder designs $\{\mathbf{V}_{jk}:\forall j,k\}$.
To highlight the role of feedback cost reduction due to CSI filtering
at the MS, we define the notion of \emph{feedback dimension} below. 
\begin{definitn}
[CSI Feedback Dimension]\label{Feedback-DimensionDefine-the}Define
the feedback dimension $D$ as the sum of the dimension of the Grassmannian
manifolds \cite{dai2008quantization} $\{\mathbb{G}(A_{jk}^{[i]},B_{jk}^{[i]}):i=1,\cdots,l_{jk};j=1,\cdots,G;k=1,\cdots,K\}$,
\begin{equation}
D=\sum_{j=1}^{G}\sum_{k=1}^{K}\sum_{i=1}^{l_{jk}}A_{jk}^{[i]}(B_{jk}^{[i]}-A_{jk}^{[i]}).\label{eq:definition_feedback_dimensioin}
\end{equation}
 \hfill \QED
\end{definitn}

\begin{remrk}
[Interpretation of CSI Feedback Dimension]The feedback dimension
in Def. \ref{Feedback-DimensionDefine-the} is a first order measure
of CSI feedback cost in MIMO cellular networks because it isolates
the contribution of CSI feedback reduction due to \emph{CSI feedback
filtering} from \emph{CSI quantization}. First, a Grassmannian manifold
of dimension $D$ is locally homeomorphic%
\footnote{The locally homeomorphic relationship between a Grassmannian $\mathbb{G}$
with dimension $D$ and $\mathbb{C}^{D\times1}$ means: there exist
a mapping $f:\mathbb{G}\rightarrow\mathbb{C}^{D\times1}$, such that
for any point $x\in\mathbb{G}$, there exists an open set $U\subseteq\mathbb{G}$
containing $x$ and the image $f(U)$ is open in $\mathbb{C}^{D\times1}$
\cite{hirsch1976differential}.%
} to $\mathbb{C}^{D\times1}$. Intuitively, this means that a Grassmannian
manifold of dimension $D$ locally looks like the $D$-dimensional
Euclidean space and a feedback dimension $D$ means that $D$ scalars
are required to feedback to the BS side. Second, the feedback dimension
is also directly proportional to the total number of bits allocated
for CSI feedback in MIMO cellular networks. As in Theorem 5 in Section
V, we demonstrate that with a total number of CSI feedback bits $D\log\mbox{SNR}$,
it is sufficient to support certain DoF in MIMO cellular networks.
\end{remrk}

\subsection{Interference Alignment under Partial CSI Feedback}

One commonly adopted IA formulation in MIMO cellular networks is to
find out the precoder and decorrelator solutions $\{\mathbf{U}_{jk},\mathbf{V}_{jk}\}$
based on the full CSIT knowledge, such that the following set of conditions
can be satisfied:
\begin{eqnarray}
\textrm{rank}(\mathbf{U}_{jk}^{\dagger}\mathbf{H}_{jk,j}\mathbf{V}_{jk})=d,\forall j,k;\label{eq:MIA_0}\\
\mathbf{U}_{jk}^{\dagger}\mathbf{H}_{jk,j}\mathbf{V}_{jp}=\mathbf{0},\;\forall j,k\neq p; &  & \mbox{(intracell interference nulling)}\label{eq:MIA_1}\\
\mathbf{U}_{jk}^{\dagger}\mathbf{H}_{jk,i}\left[\begin{array}{ccc}
\mathbf{V}_{i1} & \cdots & \mathbf{V}_{iK}\end{array}\right]=\mathbf{0},\forall j,k,i\neq j. &  & \mbox{(intercell interference nulling)}\label{eq:MIA_2}
\end{eqnarray}
However, in the above formulation of IA constraints (\ref{eq:MIA_0})-(\ref{eq:MIA_2}),
the precoders $\{\mathbf{V}_{jk}:\forall j,k\}$ serve to null both
the intracell interference in (\ref{eq:MIA_1}) and intercell interference
in (\ref{eq:MIA_2}). As such, this formulation makes it hard to find
out the CSI dependencies of the precoders $\{\mathbf{V}_{jk}:\forall j,k\}$
\cite{ruan2012interference}. Consequently, it is difficult to know
which part of CSI can be filtered out while still achieving the IA
(\ref{eq:MIA_0})-(\ref{eq:MIA_2}). To simplify the interference
nulling structure, we consider using a \emph{two-stage precoding}
structure for the precoders $\{\mathbf{V}_{jk}\}$.
\begin{definitn}
[Two Stage Precoding at the BS]\emph{\label{Two-Stage-Precoding}Two
stage precoding} is applied at each of the BSs $\{1,\cdots,G\}$,
i.e., the precoder $\mathbf{V}_{jk}$ is given by $\mathbf{V}_{jk}=\mathbf{T}_{j}\mathbf{V}_{jk}^{s}$,
where the semi-unitary matrix $\mathbf{T}_{j}\in\mathbb{U}(N,Kd)$,
$N\geq Kd$, is the \emph{outer precoder} for intercell interference
nulling and $\mathbf{V}_{jk}^{s}\in\mathbb{U}(Kd,d)$ is the \emph{inner
precoder} for intracell interference nulling between the MSs. \hfill \QED
\end{definitn}

With two stage precoding, the IA constraints (\ref{eq:MIA_0})-(\ref{eq:MIA_2})
can be reformulated as: Find out the outer precoders $\{\mathbf{T}_{i}\in\mathbb{U}(N,Kd):\forall i\}$,
inner precoders $\{\mathbf{V}_{jk}^{s}\in\mathbb{C}^{Kd\times d}:\forall j,k\}$
and decorrelators $\{\mathbf{U}_{jk}:\forall j,k\}$ based on the
full CSIT knowledge such that:
\begin{eqnarray}
\textrm{rank}(\mathbf{U}_{jk}^{\dagger}\mathbf{H}_{jk,j}\mathbf{T}_{j}\mathbf{V}_{jk}^{s})=d,\forall j,k;\label{eq:MIA_3}\\
\mathbf{U}_{jk}^{\dagger}\mathbf{H}_{jk,j}\mathbf{T}_{j}\mathbf{V}_{jp}^{s}=\mathbf{0},\;\forall j,k\neq p; &  & \mbox{\mbox{(intracell interference nulling)}}\label{eq:MIA_4}\\
\mathbf{U}_{jk}^{\dagger}\mathbf{H}_{jk,i}\mathbf{T}_{i}=\mathbf{0},\forall j,k,i\neq j. &  & \mbox{\mbox{(intercell interference nulling)}}\label{eq:MIA_5}
\end{eqnarray}

As can be seen above, the outer precoders $\{\mathbf{T}_{i}\}$ serve
to null the intercell interference only (as in (\ref{eq:MIA_5})),
and based on the outer precoders $\{\mathbf{T}_{i}\}$, the inner
precoders $\{\mathbf{V}_{ip}^{s}\}$ serves to null the intracell
interference only (as in (\ref{eq:MIA_4})). This decoupled interference
nulling structure enables us to find how the precoders adapt to the
CSI and may guide us to design efficient CSI feedback reduction schemes.
Note that the two formulations of IA constraints, i.e., (\ref{eq:MIA_0})-(\ref{eq:MIA_2})
and (\ref{eq:MIA_3})-(\ref{eq:MIA_5}), are in fact equivalent. 
\begin{lemma}
[Equivalent IA Formulation]\label{Equivalent-IA-Formulation}With
full CSIT, there exist $\{\mathbf{U}_{jk},\mathbf{V}_{jk}\}$ satisfying
constraints (\ref{eq:MIA_0})-(\ref{eq:MIA_2}) iff there exist $\{\mathbf{T}_{i}\}$,
$\{\mathbf{V}_{jk}^{s}\}$, $\{\mathbf{U}_{jk}\}$ satisfying (\ref{eq:MIA_3})-(\ref{eq:MIA_5}).
\end{lemma}

Based on the new IA constraints (\ref{eq:MIA_3})-(\ref{eq:MIA_5}),
we then investigate how the CSI can be filtered to reduce the CSI
feedback dimension. In the literature, there are some CSI feedback
designs \cite{thukral2009interference,krishnamachari2009interference,rao2012limited}
that feedback the full CDI, i.e., $F_{jk}=\left(\begin{array}{ccc}
\cdots, & \mathbb{P}(\mathbf{H}_{jk,i}), & \cdots\end{array}\right)_{\forall i}$, $\forall j,k$, which correspond to a CSI feedback dimension of
$G^{2}K(MN-1)$. By using the two stage precoding structure, we show
in Example 1 and 2 below that the IA constraints (\ref{eq:MIA_3})-(\ref{eq:MIA_5})
can still be achieved with substantially reduced feedback cost.

\begin{figure}
\begin{centering}
\includegraphics[scale=0.8]{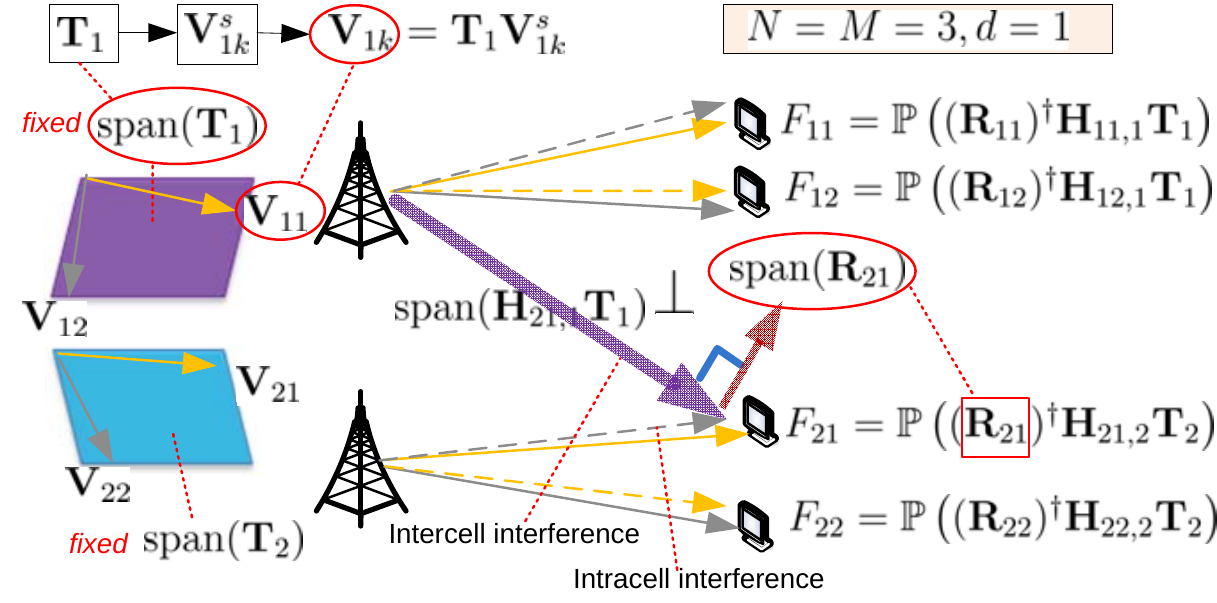}
\par\end{centering}

\caption{\label{fig:Example-1}Toy Example 1: Two-stage precoding with \emph{fixed}
outer precoder at the BSs can help to reduce the CSI feedback dimension
for IA. BS 1 has\emph{ fixed} outer precoder\textcolor{blue}{{} }$\mathbf{T}_{1}\in\mathbb{U}(3,2)$
and the $(2,1)$-th, $(2,2)$-th MSs can cancel the intercell interference
by designing the decorrelator $\mathbf{U}_{21}=\mathbf{R}_{21}$,
$\mathbf{U}_{22}=\mathbf{R}_{22}\in\mathbb{U}(3,1)$ to be orthogonal
to $\textrm{span}(\mathbf{H}_{21,1}\mathbf{T}_{1})$, $\textrm{span}(\mathbf{H}_{22,1}\mathbf{T}_{1})$
respectively (similarly for BS 2).}

\end{figure}

\begin{example}
[Two Stage Precoding with Fixed Outer Precoders]Consider a MIMO
cellular network as illustrated in Figure \ref{fig:Example-1}. Suppose
BS 1, 2 use \emph{fixed} outer precoder $\mathbf{T}_{1},\mathbf{T}_{2}\in\mathbb{U}(3,2)$.
The intercell interference space at the (2,1)-th MS is given by $\textrm{span}(\mathbf{H}_{21,1}\mathbf{T}_{1})$.
This can be cancelled by choosing a decorrelator at the $(2,1)$-th
MS as: $\mathbf{U}_{21}=\mathbf{R}_{21}\in\mathbb{U}(3,1)$, where
$\mathbf{R}_{21}$ is orthogonal to the intercell interference, i.e.,
$(\mathbf{R}_{21})^{\dagger}\mathbf{H}_{21,1}\mathbf{T}_{1}=\mathbf{0}$.
The remaining freedom at BS 2 are the inner precoders $\{\mathbf{V}_{21}^{s},\mathbf{V}_{22}^{s}\}$
which are designed to cancel the intracell interference, i.e. $(\mathbf{R}_{22})^{\dagger}\mathbf{H}_{22,2}\mathbf{T}_{2}\mathbf{V}_{21}^{s}=\mathbf{0}$,
$(\mathbf{R}_{21})^{\dagger}\mathbf{H}_{21,2}\mathbf{T}_{2}\mathbf{V}_{22}^{s}=\mathbf{0}$.
As such, the BS 2 only needs to know $F_{21}=\mathbb{P}\left((\mathbf{R}_{21})^{\dagger}\mathbf{H}_{21,2}\mathbf{T}_{2}\right)$,
$F_{22}=\mathbb{P}\left((\mathbf{R}_{22})^{\dagger}\mathbf{H}_{22,2}\mathbf{T}_{2}\right)$
to compute the inner precoders (similarly for BS 1). Hence, using
a feedback function $F_{1k}=\mathbb{P}((\mathbf{R}_{1k})^{\dagger}\mathbf{H}_{1k,1}\mathbf{T}_{1})$,
$F_{2k}=\mathbb{P}((\mathbf{R}_{2k})^{\dagger}\mathbf{H}_{2k,2}\mathbf{T}_{2})$,
$\forall k=1,2$, the IA conditions in (\ref{eq:MIA_3})-(\ref{eq:MIA_5})
can be achieved with a feedback dimension of $4\times\left(2\times1-1\right)=4$
instead of $4\times2\times(3\times2-1)=40$ in full CDI feedback. 
\end{example}

\begin{example}
[CSI Submatrix Feedback]Consider a MIMO cellular network with $G=2$
BSs and $K=3$ MSs for each BS. The BS and MS have $N=5$, $M=3$
antennas respectively and $d=1$ data stream is transmitted for each
MS. Suppose the CSI filtering functions at the MSs are given by: $F_{jk}=\left(\mathbb{P}\left(\mathbf{H}_{jk,1}^{s}\right),\mathbb{P}\left(\mathbf{H}_{jk,2}^{s}\right)\right)$,
$\forall j,k$, where $\mathbf{H}_{jk,i}^{s}$ is the $2\times5$
upper submatrix of $\mathbf{H}_{jk,i}\in\mathbb{C}^{3\times5}$, i.e.,
$\mathbf{H}_{jk,i}^{s}=[\begin{array}{cc}
\mathbf{I}_{2} & \mathbf{0}_{2\times1}\end{array}]\mathbf{H}_{jk,i}$ $\forall j,k,i$. Based on this CSI feedback $\{F_{jk}\}$, the IA
conditions (\ref{eq:MIA_3})-(\ref{eq:MIA_5}) can be achieved%
\footnote{Note that a $G=2$, $K=3$, $N=5$, $M=2$, $d=1$ MIMO cellular network
with full CSIT is IA feasible \cite{liu2013feasibility}. %
} by using the first 2 antennas at the MSs only with conventional IA
design \cite{liu2013feasibility}. As a result, the feedback dimension
is only $6\times2\times(2\times5-1)=108$ compared with $6\times2\times(3\times5-1)=168$
under full CDI feedback. 
\end{example}

Note that the strategy described in Example 1 is first mentioned in
\cite{suh2011downlink} and it can be generalized to MIMO cellular
networks with a \emph{subset} of BSs to have \emph{fixed} outer precoders. 
\begin{remrk}
Examples 1 and 2 are only simple toy examples to illustrate two effective
CSI feedback filtering policy (two-stage precoding with fixed outer
precoders and CSI submatrix feedback respectively) to reduce the CSI
feedback dimension. While these are trivial in these simple toy examples,
the challenge is to have a CSI feedback filtering solution that embrace
\emph{both} strategies to minimize the CSI feedback dimension for
general topology under DoF and IA feasibility constraints.
\end{remrk}

In the following, we shall formally give the structural form for the
CSI filtering function $F_{jk}$ that embraces the above two policies.
We first partition the BSs into two sets and define CSI \emph{submatrix}
feedback as follows.

\begin{definitn}
[Partitioning of BSs]The group of BSs \{$1$,\ldots{},$G$\} are
partitioned into two subsets, namely the type-I BSs, $\mathbb{B}_{g}^{I}=\{1,\cdots,g\}$
and the type-II BSs, $\mathbb{B}_{g}^{II}=\{g+1,\cdots,G\}$. The
type-II BSs use \emph{fixed} outer precoder $\mathbf{T}_{i}^{II}\in\mathbb{U}(N,Kd)$,
$i\in\mathbb{B}_{g}^{II}$. 
\end{definitn}

\begin{definitn}
[CSI Submatrix Feedback]The CSI submatrices $\{\mathbf{H}_{jk,i}^{s}\}$
are considered for CSI filtering feedback, where $\{\mathbf{H}_{jk,i}^{s}\}$
correspond to the CSI on the first $m_{jk}\leq M$ antennas at the
$(j,k)$-th MS, $\forall j,k$, the first $n_{i}\leq N$ antennas
at the $i$-th BS, $i\in\mathbb{B}_{g}^{I}$ , and degenerated $n_{i}=N$
antennas at the $i$-th BS, $i\in\mathbb{B}_{g}^{II}$. That is:
\begin{equation}
\mathbf{H}_{jk,i}^{s}=\begin{cases}
\left[\begin{array}{cc}
\mathbf{I}_{m_{jk}} & \mathbf{0}\end{array}\right]\mathbf{H}_{jk,i}\left[\begin{array}{cc}
\mathbf{I}_{n_{i}} & \mathbf{0}\end{array}\right]^{T}, & \forall j,k,i\in\mathbb{B}_{g}^{I}\\
\left[\begin{array}{cc}
\mathbf{I}_{m_{jk}} & \mathbf{0}\end{array}\right]\mathbf{H}_{jk,i}, & \forall j,k,i\in\mathbb{B}_{g}^{II}
\end{cases}.\label{eq:sub_matrix}
\end{equation}
\hfill \QED
\end{definitn}

Note that $\{m_{jk}:\forall j,k\}$, $\{n_{i}:\forall i\in\mathbb{B}_{g}^{I}\}$
characterizes the size%
\footnote{Note that when $m_{jk}\ne M$ or $n_{i}\neq N$ in $\mathcal{L}$,
instead of directly selecting the upper left $m_{jk}\times n_{i}$
submatrix as in (\ref{eq:sub_matrix}), there is in fact \emph{extra
space} of carefully selecting the $m_{jk}$ or $n_{i}$ effective
antennas to improve the direct link power gain. However, in this paper,
we are more interested in the tradeoff between the \emph{first-order}
DoF performance and the CSI feedback cost, and note the possible power
gain mentioned will not affect the performance in the DoF sense \cite{cadambe2008interference}.
Therefore, to better illustrate the insights, we consider a simple
effective antenna reduction scheme as in (\ref{eq:sub_matrix}) and
focus on the feedback dimension reduction of $\mathcal{L}$. %
} of the CSI submatrices $\{\mathbf{H}_{jk,i}^{s}\}$. Denote $\mathbb{N}^{r}(\cdot)$
as the left null space, i.e., $\mathbb{N}^{r}(\mathbf{A})=\{\mathbf{u}\mid\mathbf{u}^{\dagger}\mathbf{A}=\mathbf{0}\}$.
Based on the above two definitions, we have the following definition
on the CSI filtering functions $\{F_{jk}\}$. 
\begin{definitn}
[Structural Form of $F_{jk}$]\label{Structural-Form-of-CSI-Feedback}The
CSI filtering functions $F_{jk}(\mathcal{H}_{jk})$ in (\ref{eq:feedback-info})
are given by
\begin{equation}
F_{jk}(\mathcal{H}_{jk})=\left(\cdots,\mathbb{P}\left(\mathbf{H}_{jk,i}^{e}\right),\cdots\right)_{i\in\mathbb{B}_{g}^{I}\bigcup\{j\}},\forall j,k;\label{eq:F_function}
\end{equation}
\begin{equation}
\mathbf{H}_{jk,i}^{e}=\begin{cases}
(\mathbf{R}_{jk})^{\dagger}\mathbf{H}_{jk,i}^{s}\in\mathbb{C}^{A_{jk}\times n_{i}}, & \forall j,k,i\in\mathbb{B}_{g}^{I}\\
(\mathbf{R}_{jk})^{\dagger}\mathbf{H}_{jk,i}^{s}\mathbf{T}_{j}^{II}\in\mathbb{C}^{A_{jk}\times Kd}, & \forall k,j\in\mathbb{B}_{g}^{II}
\end{cases}.\label{eq:effective_CSi}
\end{equation}
where $\mathbf{H}_{jk,i}^{e}$ denotes the effective CSI, $\mathbf{R}_{jk}\in\mathbb{U}(m_{jk},A_{jk})$
is a semi-unitary matrix that defines%
\footnote{We define $\mathbf{R}_{jk}=\mathbf{I}$ when $\mathbb{B}_{g}^{II}\backslash\{j\}=\emptyset$. %
} the left null space of the intercell interference from all type-II
BSs at the $(j,k)$-th MS: 
\begin{equation}
\textrm{span}(\mathbf{R}_{jk})=\mathbb{N}^{r}\left(\left[\begin{array}{ccc}
\cdots & \mathbf{H}_{jk,i}^{s}\mathbf{T}_{i}^{II} & \cdots\end{array}\right]_{i\in\mathbb{B}_{g}^{II}\backslash\{j\}}\right),\forall j,k;\label{eq:s_r_expression}
\end{equation}
\begin{equation}
A_{jk}=m_{jk}-\sum_{i\in\mathbb{B}_{g}^{II}\backslash\{j\}}Kd,\,\forall j,k.\label{eq:A_definition}
\end{equation}
\hfill \QED
\end{definitn}

Note that there is no need to feedback the intercell cross link CSIs
$\left\{ \mathbf{H}_{jk,i}^{s}:\forall j,k,i\in\mathbb{B}_{g}^{II}\backslash\{j\}\right\} $
because the intercell interference from type-II BSs can be canceled
by setting the decorrelator $\mathbf{U}_{jk}$ to be in the subspace
spanned by $\mathbf{R}_{jk}$. The above feedback structure in Def.
\ref{Structural-Form-of-CSI-Feedback} corresponds to the tuple $\mathbb{\mathcal{H}}_{jk}^{fed}=F_{jk}(\mathcal{H}_{jk})\in\mathbb{G}\left(1,B_{jk,1}\right)\times\cdots\mathbb{G}\left(1,B_{jk,l_{jk}}\right)$,
where the length $l_{jk}=|\mathbb{B}_{g}^{I}\bigcup\{j\}|$ and 
\begin{equation}
B_{jk,i}=\begin{cases}
KdA_{jk}, & i=l_{jk},j\in\mathbb{B}_{g}^{II}\\
n_{i}A_{jk}, & j\in\mathbb{B}_{g}^{I}\mbox{ or }i<l_{jk}
\end{cases}\label{eq:B_definition}
\end{equation}
as in (\ref{eq:feedback-info}). Based on the above, we define the
notion of CSI \emph{feedback profile}, which gives a parametrization
of $\{F_{jk}\}$. 
\begin{definitn}
[Feedback Profile of $\{F_{jk}\}$]\label{Feedback-Profile-of}Define
the feedback profile of $\{F_{jk}\}$ as a set of parameters:
\begin{equation}
\mathcal{L}=\left\{ \{m_{jk}:\forall j,k\},g,\{n_{i}:\forall i\in\mathbb{B}_{g}^{I}\}\right\} .\label{eq:feedback_topology}
\end{equation}
\hfill \QED
\end{definitn}

Note that $m_{jk}$ and $n_{i}$ in $\mathcal{L}$ control the size
of the CSI submatrices to feedback and $g=|\mathbb{B}_{g}^{I}|$ is
the number of the type-I BS. In fact, there is a one-to-one correspondence
between the feedback profile $\mathcal{L}$ and the feedback function
in (\ref{eq:F_function}). Note that the proposed feedback profile
$\mathcal{L}$ embraces Example 1, 2 with the corresponding $\mathcal{L}=\{\{m_{jk}=3,\forall j,k\},g=0\}$
in Example 1 and $\mathcal{L}=\{\{m_{jk}=2:\forall j,k\},g=2,\{n_{i}=5:i\leq2\}\}$
in Example 2. Furthermore, it also includes some existing works as
special cases:
\begin{itemize}
\item \textbf{Special Case I} \emph{(Full CDI Feedback):} When $\mathcal{L}=\{\{m_{jk}=M,\forall j,k\},g=G,\{n_{i}=N:\forall i\}\}$,
$\mathcal{L}$ will be reduced to conventional full CDI feedback in
\cite{thukral2009interference,krishnamachari2009interference,rao2012limited}. 
\item \textbf{Special Case II} \emph{(Zero-forcing IA Feedback)}: When $G=2$,
$M=N=K+1$, $d=1$ and $\mathcal{L}=\{\{m_{jk}=M,\forall j,k\},g=0\}$
(all BSs are typ-II BSs), $\mathcal{L}$ will be reduced to the feedback
scheme in \cite{suh2011downlink} (Example 1 corresponds to one such
example). 
\end{itemize}

For a given feedback profile $\mathcal{L}$, the total feedback dimension
is given by,
\begin{equation}
D(\mathcal{L})=\sum_{j=1}^{G}\sum_{k=1}^{K}\sum_{i=1}^{g}(n_{i}A_{jk}-1)+\sum_{j=g+1}^{G}\sum_{k=1}^{K}(KdA_{jk}-1).\label{eq:sum_feedback_dimension_expression}
\end{equation}

Next, we discuss IA constraints under the proposed CSI filtering $\mathcal{L}$,
to achieve $d$ data streams for each MS in the following.
\begin{constraints}
[IA under $\mathcal{L}$]\label{IA-design-based-general-L}Given
the CSI feedback profile $\mathcal{L}$ and the outer precoders $\{\mathbf{T}_{i}^{II}\in\mathbb{U}(N,Kd):i\in\mathbb{B}_{g}^{II}\}$
for the type-II BSs, find the outer precoders $\{\mathbf{T}_{i}^{I}\in\mathbb{U}(N,Kd):i\in\mathbb{B}_{g}^{I}\}$
for type-I BSs, the inner precoders $\{\mathbf{V}_{jk}^{s}\in\mathbb{U}(Kd,d):\forall j,k\}$
for all BSs and decorrelators $\{\mathbf{U}_{jk}\}$ for all MSs,
such that: 
\begin{eqnarray}
\textrm{rank}(\mathbf{U}_{jk}^{\dagger}\mathbf{H}_{jk,j}\mathbf{T}_{j}\mathbf{V}_{jk}^{s})=d,\forall j,k;\label{eq:MIA-1}\\
\mathbf{U}_{jk}^{\dagger}\mathbf{H}_{jk,j}\mathbf{T}_{j}\mathbf{V}_{jp}^{s}=\mathbf{0},\forall j,k\neq p; & \mbox{} & \mbox{(intracell IA constraints)}\label{eq:IA_condition-1}\\
\mathbf{U}_{jk}^{\dagger}\mathbf{H}_{jk,i}\mathbf{T}_{i}=\mathbf{0},\forall j,k,i\neq j; &  & \mbox{(\mbox{intercell IA} constraints)}\label{eq:IA_condition_2}
\end{eqnarray}
\begin{equation}
\begin{array}{cc}
\begin{array}{c}
\left\{ \mathbf{T}_{j}^{I}:i\in\mathbb{B}_{g}^{I}\right\} ,\{\mathbf{V}_{jk}^{s}:\forall j,k\}\mbox{ can only be }\\
\mbox{ adaptive to \ensuremath{\{F_{jk}(\mathcal{H}_{jk}):\forall j,k\}\;}according to \ensuremath{\mathcal{L}}.}
\end{array} & \begin{array}{c}
(\mbox{CSI knowledge}\\
\mbox{constraint)}
\end{array}\end{array}\label{eq:IA_hidden}
\end{equation}
where $\mathbf{T}_{j}=\mathbf{T}_{j}^{I}$, $j\in\mathbb{B}_{g}^{I}$
and $\mathbf{T}_{j}=\mathbf{T}_{j}^{II}$, $j\in\mathbb{B}_{g}^{II}$
for notation convenience. \hfill \QED
\end{constraints}

Note that (\ref{eq:MIA-1})-(\ref{eq:IA_condition_2}) refers to the
\emph{IA constraints} and (\ref{eq:IA_hidden}) refers to the \emph{CSI
knowledge constraint}. Specifically, compared with conventional IA
with full CSIT in (\ref{eq:MIA_3})-(\ref{eq:MIA_5}), there are two
unique challenges, namely the \emph{CSI knowledge} and \emph{feasibility},
associated with Constraints I under partial CSIT knowledge. First,
the CSI knowledge constraint is an implicit constraint which is difficult
to handle. Second, adjusting the feedback profile $\mathcal{L}$ may
reduce the CSI feedback dimension $D(\mathcal{L})$ in (\ref{eq:sum_feedback_dimension_expression})
but the IA constraints may no longer be feasible. The following summarizes
the challenges we face.\vspace{-1cm}

\begin{center}
\framebox{\begin{minipage}[t]{1\columnwidth}%
Challenge 1: Adjust the feedback profile $\mathcal{L}$ so as to minimize
the feedback dimension $D(\mathcal{L})$ subject to Constraints 1.%
\end{minipage}}
\par\end{center}

\section{IA Feasibility Conditions Under a Given Feedback Profile $\mathcal{L}$}

In this section, we shall investigate Constraints 1 and find out the
requirements on $\mathcal{L}$ to make the IA problem in Constraint
\ref{IA-design-based-general-L} feasible, i.e., for what kind of
CSI feedback profile $\mathcal{L}$, Constraints \ref{IA-design-based-general-L}
can have feasible solutions $\{\mathbf{T}_{j}^{I}\}$, $\{\mathbf{V}_{jk}^{s},\mathbf{U}_{jk}\}$.
We further derive the corresponding IA transceiver solutions $\{\mathbf{T}_{j}^{I}\}$,
$\{\mathbf{V}_{jk}^{s},\mathbf{U}_{jk}\}$ to satisfy the conditions
in Constraints \ref{IA-design-based-general-L} for a given feedback
profile $\mathcal{L}$.

\subsection{IA Constraints Transformation}

To investigate (\ref{eq:IA_hidden}) in Constraints \ref{IA-design-based-general-L},
we shall first have a better understanding on how to utilize the partial
CSI knowledge $\{F_{jk}\}$. Specifically, the information available
at BS $j$ from the feedback CSI $\{F_{jk}:\forall k\}$ is denoted
by the set of matrices $\mathbb{H}_{j}$,
\begin{equation}
\mathbb{H}_{j}=\left\{ \mathbf{\tilde{H}}_{jk,i}^{e}=a_{jk,i}\mathbf{H}_{jk,i}^{e}:\forall k,i\in\mathbb{B}_{g}^{I}\bigcup\{j\}\right\} \label{eq:available_CSI}
\end{equation}
where $\{a_{jk,i}\}$ are some%
\footnote{As an example, one common approach \cite{dai2008quantization} to
feedback $\mathbb{P}(\mathbf{H})=\{a\mathbf{H}:a\in\mathbb{C}\},$
$\mathbf{H}\in\mathbb{C}^{A\times B}$ is to feedback the unitary
vector $\frac{1}{||\mathbf{H}||}\textrm{vec}(\mathbf{H})$. In this
case, the scalar $a=\frac{1}{||\mathbf{H}||}$.%
} non-zero scalars. Based on $\{\mathbb{H}_{j}\}$, we study Constraints
\ref{IA-design-based-general-L}.

\vspace{-1cm}

\begin{center}
\framebox{\begin{minipage}[t]{1\columnwidth}%
Challenge 2: Constraints \ref{IA-design-based-general-L} is difficult
because 1) the conditions (\ref{eq:MIA-1}) and (\ref{eq:IA_condition-1})
are coupled as $\{\mathbf{H}_{jk,j}\}$ act as both the direct link
in (\ref{eq:MIA-1}) and the cross link in (\ref{eq:IA_condition-1});
2) the \emph{CSI knowledge constraint} (\ref{eq:IA_hidden}) requires
that the precoders can only be designed based on the partial CSI knowledge
$\{\mathbb{H}_{j}\}$.%
\end{minipage}}
\par\end{center}

We first introduce an equivalent IA constraint transformation, which
can explicitly handle the CSI knowledge constraint and the coupling
issues.

\begin{constraints}
[IA Constraint Transformation under $\mathcal{L}$]\label{Transformed-Problem-II}Find
$\mathbf{\tilde{T}}_{i}^{I}\in\mathbb{U}(n_{i},Kd)$, $n_{i}\geq Kd$,
$i\in\mathbb{B}_{g}^{I}$, and $\mathbf{\tilde{U}}_{jk}\in\mathbb{U}(A_{jk},d)$,
$A_{jk}\geq d$, $\forall j,k$, satisfying the following equations:
\begin{equation}
(\mathbf{\tilde{U}}_{jk})^{\dagger}\mathbf{\tilde{H}}_{jk,i}^{e}\mathbf{\tilde{T}}_{i}^{I}=\mathbf{0},\,\forall j,k,i\in\mathbb{B}_{g}^{I}\backslash\{j\}.\label{eq:transformed_II}
\end{equation}
\hfill \QED
\end{constraints}

Note that constraint in (\ref{eq:transformed_II}) involves intercell
IA constraints from the \emph{type-I} BSs only. This is because the
intercell interference from the type-II BSs has already been cancelled
at the MSs via designing the decorrelator $\mathbf{U}_{jk}$ in the
subspace spanned by $\mathbf{R}_{jk}$. Furthermore, Constraint \ref{Transformed-Problem-II}
contains no intracell IA constraints because of the two stage precoding
structures in (\ref{eq:MIA-1})-(\ref{eq:IA_condition_2}). The equivalent
relationship between Constraints \ref{IA-design-based-general-L}
and Constraints \ref{Transformed-Problem-II} is established in the
lemma below. 
\begin{lemma}
[Equivalence of Constraints \ref{IA-design-based-general-L} and
Constraints \ref{Transformed-Problem-II} ]\label{Equivalence-of-Problem-II}Given
the CSI feedback profile $\mathcal{L}$ and the outer precoders $\{\mathbf{T}_{i}^{II}\in\mathbb{U}(N,Kd):i\in\mathbb{B}_{g}^{II}\}$
for type-II BSs:

(a): Constraints \ref{Transformed-Problem-II} is feasible iff Constraints
\ref{IA-design-based-general-L} is feasible. 

(b): If $\{\mathbf{\tilde{T}}_{i}^{I}\}$ and $\{\mathbf{\tilde{U}}_{jk}\}$
are solutions of Constraints \ref{Transformed-Problem-II}, then $\{\mathbf{T}_{i}^{I}\}$,
$\{\mathbf{V}_{jk}^{s},\mathbf{U}_{jk}\}$ given by
\begin{equation}
\mathbf{T}_{i}^{I}=\left[\begin{array}{c}
\mathbf{\tilde{T}}_{i}^{I}\\
\mathbf{0}
\end{array}\right],i\in\mathbb{B}_{g}^{I},\quad\mathbf{U}_{jk}=\left[\begin{array}{c}
\mathbf{R}_{jk}\mathbf{\tilde{U}}_{jk}\\
\mathbf{0}
\end{array}\right],\forall j,k;\label{eq:structure}
\end{equation}
\begin{equation}
\mathbf{V}_{jk}^{s}=v_{(d)}\left(\sum_{\underset{\neq k}{p=1}}^{K}\left((\mathbf{\tilde{U}}_{jp})^{\dagger}\mathbf{\tilde{H}}_{jp,j}^{e}\mathbf{\tilde{T}}_{i}^{I}\right)^{\dagger}\left((\mathbf{\tilde{U}}_{jp}\mathbf{\tilde{H}}_{jp,j}^{e}\mathbf{\tilde{T}}_{i}^{I}\right)\right),\forall k,j\in\mathbb{B}_{g}^{I};\label{eq:intra_cell_relation}
\end{equation}
\begin{equation}
\mathbf{V}_{jk}^{s}=v_{(d)}\left(\sum_{\underset{\neq k}{p=1}}^{K}\left((\mathbf{\tilde{U}}_{jp})^{\dagger}\mathbf{\tilde{H}}_{jp,j}^{e}\right)^{\dagger}\left((\mathbf{\tilde{U}}_{jp})^{\dagger}\mathbf{\tilde{H}}_{jp,j}^{e}\right)\right),\forall k,j\in\mathbb{B}_{g}^{II}\label{eq:intra_cell_relation-1}
\end{equation}
are the solutions for Constraints \ref{IA-design-based-general-L}
almost surely, where $v_{(d)}(\mathbf{A})$ is the matrix of eigenvectors
corresponding to the $d$ least eigenvalues of a Hermitian matrix
$\mathbf{A}$.
\end{lemma}

Note that the precoder solutions in (\ref{eq:structure})-(\ref{eq:intra_cell_relation-1})
automatically satisfies the \emph{CSI knowledge constraint} (\ref{eq:IA_hidden}).
Furthermore, the IA Constraints \ref{Transformed-Problem-II} contain
the intercell\emph{ }IA constraints from type-I BSs only as in (\ref{eq:transformed_II}).
Consequently, the aforementioned Challenge 2 is tackled by using Constraint
\ref{Transformed-Problem-II} and Lemma \ref{Equivalence-of-Problem-II}.

\subsection{Feasibility Conditions on $\mathcal{L}$}

Based on Lemma \ref{Equivalence-of-Problem-II} and Constraints 2,
we obtain the following necessary feasibility conditions for Constraint
1.

\begin{thm}
[Necessary Conditions for IA Feasible on $\mathcal{L}$]\label{Necessary-IA-Feasibility}If
Constraints \ref{IA-design-based-general-L} is feasible, the CSI
feedback profile $\mathcal{L}$ should satisfy: 1) $m_{jk}-\sum_{i\in\mathbb{B}_{g}^{II}\backslash\{j\}}Kd-d\geq0,\forall j,k$,
2) $N\geq Kd$, $n_{i}\geq Kd$, $i\in\mathbb{B}_{g}^{I}$, 3) $\forall\mathcal{J}_{sub}^{[r]}\subseteq\{(j,k):\forall j,k\},\mathcal{J}_{sub}^{[t]}\subseteq\mathbb{B}_{g}^{I}$,
\begin{equation}
\sum_{(j,k)\in\mathcal{J}_{sub}^{[r]}}\left(m_{jk}-\sum_{i\in\mathbb{B}_{g}^{II}\backslash\{j\}}Kd-d\right)+\sum_{i\in\mathcal{J}_{sub}^{[t]}}K(n_{i}-Kd)\geq\sum_{j\in\mathcal{J}_{sub}^{[r]}}\sum_{i\in\mathcal{J}_{sub}^{[t]}\backslash\{j\}}Kd.\label{eq:necessary_feasibility_condition}
\end{equation}

\end{thm}

For instance, if we have 0 type-I BS ($g=0$) and $m_{jk}=m$, $\forall j,k$,
in $\mathcal{L}$, then Theorem 1 requires $N\geq Kd$, $m\geq(G-1)Kd+d$
for $\mathcal{L}$ to be IA feasible (see Example 1); if we have 0
type-II BS $(g=G)$ and $m_{jk}=m$, $\forall j,k$, $n_{i}=n,$ $\forall i$,
in $\mathcal{L}$, then Theorem 1 requires $m\geq d$, $n\geq Kd$,
$m+n\geq(GK+1)d$ for $\mathcal{L}$ to be IA feasible (see Example
2). Using the max-flow theory \cite{cormen2001introduction,edmonds1972theoretical},
Theorem \ref{Necessary-IA-Feasibility} can be expressed in an alternative
way.
\begin{cor}
[Equivalent Condition]\label{Equivalent-Condition}$\mathcal{L}$
satisfies the three conditions in Theorem \ref{Necessary-IA-Feasibility}
iff $N\geq Kd$ and there exist non-negative variables $\{f_{jk,i}^{t},f_{jk,i}^{r}\}$,
$f_{jk,i}^{r}\geq0,f_{jk,i}^{t}\geq0$, $\forall j,k,i\in\mathbb{B}_{g}^{I}\backslash\{j\}$,
that satisfy 
\begin{equation}
f_{jk,i}^{r}+f_{jk,i}^{t}\geq Kd,\,\forall j,k,i\in\mathbb{B}_{g}^{I}\backslash\{j\};\label{eq:equiva1}
\end{equation}
\begin{equation}
\left(m_{jk}-\sum_{i\in\mathbb{B}_{g}^{II}\backslash\{j\}}Kd-d\right)\geq\sum_{i\in\mathbb{B}_{g}^{I}\backslash\{j\}}^{G}f_{jk,i}^{r},\:\forall j,k;\label{eq:equiva2}
\end{equation}
\begin{equation}
(n_{i}-Kd)K\geq\sum_{j\neq i}^{G}\sum_{k=1}^{K}f_{jk,i}^{t},\:\forall i\in\mathbb{B}_{g}^{I}.\label{eq:equiva3}
\end{equation}

\end{cor}

Similar to conventional IA \cite{yetis2010feasibility,ruan2012feasibility},
checking condition (\ref{eq:necessary_feasibility_condition}) in
Theorem 1 involves an exponential number of comparisons (i.e., $O(2^{KG})$).
By Corollary \ref{Equivalent-Condition}, this exponential complexity
can be reduced to a polynomial number. On the other hand, Corollary
\ref{Equivalent-Condition} also provides a constructive approach
to verify the IA feasibility conditions (i.e., construct $f_{jk,i}^{r}$,
$f_{jk,i}^{t}$ in terms of the parameters in $\mathcal{L}$ and check
the conditions (\ref{eq:equiva1})-(\ref{eq:equiva3})). 

We also have that the conditions in Theorem \ref{Necessary-IA-Feasibility}
are sufficient in the divisible cases. 
\begin{thm}
[Sufficient IA Feasibility Conditions]\label{Sufficient-IA-Feasibility}Suppose
$\mathcal{L}$ satisfies the three conditions in Theorem \ref{Necessary-IA-Feasibility}.
If $\mathcal{L}$ further satisfies $d\mid n_{i}$, $\forall i\in\mathbb{B}_{g}^{I}$,
or $Kd\mid(m_{jk}-d)$, $\forall j,k$, Constraints \ref{IA-design-based-general-L}
is feasible. 
\end{thm}

\begin{remrk}
[Backward Compatibility with Previous Results]Suppose $g=G$, $K=1$,
$m_{jk}=M$, $n_{i}=N$ in $\mathcal{L}$ (full CDI feedback). Then
the required conditions on parameter $G$, $M$, $N$, $d$ from Theorem
2 in this paper, are the same as from Corollary 3.4 of \cite{ruan2012feasibility}.
Suppose $g=G$, $m_{jk}=M$, $n_{i}=N$, $\forall j,k,i$, in $\mathcal{L}$
(full CDI feedback) and $d\mid N$, $d\mid M$. Then the required
conditions on parameter $G$, $K$, $M$, $N$, $d$ from Theorem
2 in this paper, are the same as from Theorem 2 in \cite{liu2013feasibility}. 
\end{remrk}

\subsection{Transceiver Design under $\mathcal{L}$}

In this section, we derive the IA solutions $\{\mathbf{T}_{j}^{I}\}$,
$\{\mathbf{V}_{jk}^{s},\mathbf{U}_{jk}\}$ to Constraints 1. Note
conventional IA designs \cite{ruan2012interference,gomadam2011distributed}
require full CSIT and hence can not be directly applied to Constraints
\ref{IA-design-based-general-L} which have the \emph{CSI knowledge
constraint} (\ref{eq:IA_hidden}). Specifically, we adopt the alternating
interference leakage minimization (AILM) techniques \cite{gomadam2011distributed}
and solve the equivalent Constraints \ref{Transformed-Problem-II}.
Similar to \cite{gomadam2011distributed}, we propose the following
problem to find the solutions to satisfy Constraints \ref{Transformed-Problem-II}.
\begin{problem}
[Interference Leakage Minimization]\label{Interference-Leakage-Minimizatio}
\end{problem}
\begin{eqnarray}
\min_{\{\mathbf{\tilde{T}}_{i}^{I},\mathbf{\tilde{U}}_{jk}\}} &  & I\triangleq\sum_{(j,k)}\sum_{i\in\mathbb{B}_{g}^{I}\backslash\{j\}}\textrm{tr}\left((\mathbf{\tilde{U}}_{jk})^{\dagger}\mathbf{\tilde{H}}_{jk,i}^{e}\mathbf{\tilde{T}}_{i}^{I}\left((\mathbf{\tilde{U}}_{jk})^{\dagger}\mathbf{\tilde{H}}_{jk,i}^{e}\mathbf{\tilde{T}}_{i}^{I}\right)^{\dagger}\right)\label{eq:solve_u}\\
\textrm{s.t.} &  & \mathbf{\tilde{T}}_{i}^{I}\in\mathbb{U}(n_{i},Kd),i\in\mathbb{B}_{g}^{I};\;\mathbf{\tilde{U}}_{jk}\in\mathbb{U}(A_{jk},d),\forall j,k.\nonumber 
\end{eqnarray}
\hfill \QED

Problem \ref{Interference-Leakage-Minimizatio} has closed-form optimal
$\{\mathbf{\tilde{T}}_{i}^{I}\}$ for fixed $\{\mathbf{\tilde{U}}_{jk}\}$
and closed-form optimal $\{\mathbf{\tilde{U}}_{jk}\}$ for fixed $\{\mathbf{\tilde{T}}_{i}^{I}\}$,
and hence we shall apply \emph{alternating optimization} techniques
\cite{gomadam2011distributed} to derive solutions. 

\textit{Algorithm 1 }\textit{\emph{(}}\emph{Iterative Solution to
Constraints}\textit{ \ref{Transformed-Problem-II} under $\mathcal{L}$):}
\begin{itemize}
\item \textbf{Step 1} \textbf{(}\emph{Initialization}\textbf{)}: Randomly
initialize $\mathbf{\tilde{T}}_{i}^{I}\in\mathbb{U}(n_{i},Kd)$, $\forall i\in\mathbb{B}_{g}^{I}$,
$\mathbf{\tilde{U}}_{jk}\in\mathbb{U}(A_{jk},d)$, $\forall j,k$.
\item \textbf{Step 2 (}\emph{Update}\textbf{ $\{\mathbf{\tilde{U}}_{jk}\}$)}:
Update $\mathbf{\tilde{U}}_{jk}=v_{d}\left(\sum_{i\in\mathbb{B}_{g}^{I}\backslash\{j\}}\left(\mathbf{\tilde{H}}_{jk,i}^{e}\mathbf{\tilde{T}}_{i}^{I}\right)\left(\mathbf{\tilde{H}}_{jk,i}^{e}\mathbf{\tilde{T}}_{i}^{I}\right)^{\dagger}\right)$,
$\forall j,k$.
\item \textbf{Step 3} \textbf{(}\emph{Update}\textbf{ $\{\mathbf{\tilde{T}}_{i}^{I}\}$)}:
Update $\mathbf{\tilde{T}}_{i}^{I}=v_{(Kd)}\left(\sum_{(j,k)}\left(\mathbf{\tilde{H}}_{jk,i}^{e}\mathbf{\tilde{U}}_{jk}\right)\left(\mathbf{\tilde{H}}_{jk,i}^{e}\mathbf{\tilde{U}}_{jk}\right)^{\dagger}\right)$,
$\forall i\in\mathbb{B}_{g}^{I}$.
\item Repeat \textbf{Step 2} and \textbf{Step 3} until the value of $I$
in (\ref{eq:solve_u}) converges. \hfill \QED
\end{itemize}

Note that based on the converged solution of $\{\mathbf{\tilde{T}}_{i}^{I}\}$
and $\{\mathbf{\tilde{U}}_{jk};\}$ from Algorithm 1, we can obtain
the overall solutions $\{\mathbf{T}_{j}^{I}\}$ $\{\mathbf{V}_{jk}^{s},\mathbf{U}_{jk}\}$
to Constraints \ref{IA-design-based-general-L} by using Lemma \ref{Equivalence-of-Problem-II}.
\begin{remrk}
[Characterization of Algorithm 1]Note Algorithm 1 can automatically
adapt to the \emph{partial CSI knowledge} constraint (\ref{eq:IA_hidden}).
On the other hand, the value of $I$ converges in Algorithm 1 because:
1) the total interference leakage $I$ in (\ref{eq:solve_u}) is monotonically
decreasing in the alternating updates of Step 2 and Step 3; 2) $I$
is non-negative so that $I$ is bounded below. However, the convergence
to global optimality is not guaranteed due to the nonconvexity of
Problem \ref{Interference-Leakage-Minimizatio} \cite{gomadam2011distributed}.
Note that if the total interference leakage $I$ at the converged
point is 0, then the converged solution is a feasible solution to
Constraints \ref{Transformed-Problem-II}. Furthermore, from extensive
simulations, it is observed that the converged value of $I$ is always
0 when Constraints \ref{Transformed-Problem-II} is feasible (similar
to conventional AILM works \cite{gomadam2011distributed,ruan2012interference,gonzalez2012general}).
\end{remrk}

\subsection{Implementation Consideration}

In this section, we give a summary on how to implement the proposed
IA scheme with partial CSI feedback $\mathcal{L}$ in MIMO cellular
networks. 

\textit{Algorithm 2 }\textit{\emph{(}}\emph{Implementation of Proposed
IA Scheme}\textit{ under $\mathcal{L}$):}
\begin{itemize}
\item \textbf{Step 1} \textbf{\emph{(}}\emph{CSI Observation}\textbf{\emph{)}}:
The $(j,k)$-th MS observes the local CSI $\mathcal{H}_{jk}=(\mathbf{H}_{jk,1},\mathbf{H}_{jk,2},\cdots\mathbf{H}_{jk,G})$,
$\forall j,k$. 
\item \textbf{Step 2} (\emph{Partial CSI Feedback under} $\mathcal{L}$):
The $(j,k)$-th MS feedbacks the filtered CSI generated by $\mathbb{\mathcal{H}}_{jk}^{fed}=F_{jk}(\mathcal{H}_{jk})$
to BS $j$, where $F_{jk}$ is the CSI filtering function as in Definition
\ref{Structural-Form-of-CSI-Feedback} according to feedback profile
$\mathcal{L}$. 
\item \textbf{Step 3 (}\emph{Transceiver Computation}\textbf{)}: BS $j$
obtains $\mathbb{H}_{j}$ in (\ref{eq:available_CSI}) from the feedback
$\{\mathbb{\mathcal{H}}_{jk}^{fed}:\forall k\}$. One BS collects
the $\{\mathbb{H}_{j}:\forall j\}$ from other BSs through the backhaul
and computes $\{\mathbf{\tilde{T}}_{i}^{I}:i\in\mathbb{B}_{g}^{I}\}$,
$\{\mathbf{\tilde{U}}_{jk};\forall j,k\}$ according to Algorithm
1 in a centralized manner. 
\item \textbf{Step 4} \textbf{(}\emph{Transceiver Distribution}\textbf{)}:
The BS mentioned in Step 3 distributes the obtained $\mathbf{\tilde{T}}_{j}^{I},\{\mathbf{\tilde{U}}_{jk}:\forall j,k\}$
to BS $j$ for $j\in\mathbb{B}_{g}^{I}$ and $\{\mathbf{\tilde{U}}_{jk}:\forall k\}$
to BS $j$ for $j\in\mathbb{B}_{g}^{II}$. BS $j$ forward $\mathbf{\tilde{U}}_{jk}$
to the $(j,k)$-th MS, $\forall j,k$. 

\begin{itemize}
\item BS $j$ uses $\mathbf{T}_{j}^{I}$ as the outer precoder for type-I
BSs, $\mathbf{V}_{jk}^{s}$ as the inner precoder for the $(j,k)$-th
MS designed via equations (\ref{eq:structure})-(\ref{eq:intra_cell_relation-1})
in Lemma \ref{Equivalence-of-Problem-II}.
\item The $(j,k)$-th MS uses $\mathbf{U}_{jk}$ as the decorrelator designed
via equation (\ref{eq:structure}) in Lemma \ref{Equivalence-of-Problem-II}.\hfill \QED
\end{itemize}

\end{itemize}

\section{Feedback Dimension Minimization and Asymptotic Optimal Feedback Profile}

\subsection{Problem Formulation}

In this section, we solve Challenge 1 by solving the following problem
of CSI feedback dimension minimization subject to the requirement
of IA DoFs (Constraints 1) under partial CSI feedback $\mathcal{L}$
in MIMO cellular networks. 
\begin{problem}
[Feedback Dimension Minimization]\label{Feedback-Dimension-Optimization-2}
\begin{eqnarray}
\min_{\mathcal{L}} &  & D(\mathcal{L})\label{eq:objective-2}\\
\textrm{s.t.} &  & n_{i}\leq N,\forall i\in\mathbb{B}_{g}^{I},\, m_{jk}\leq M,\forall j,k;\label{eq:con_1-2}\\
 &  & 0\leq g\leq G,\quad g,n_{i},m_{jk}\in\mathbb{Z};\\
 &  & \textrm{Constraints 1 under \ensuremath{\mathcal{L}}}.\label{eq:con_1-4}
\end{eqnarray}
\hfill \QED
\end{problem}

Note that Problem 2 is an offline optimization where we try to find
the optimal feedback profile $\mathcal{L}^{*}$ to minimize the feedback
dimension $D(\mathcal{L})$ so that the BS can still deliver $d$
data streams to each MS in the MIMO cellular network with the given
antenna configurations. Note that the Constraints 1 in (\ref{eq:con_1-4})
is an \emph{implicit constraint} on $\mathcal{L}$ and the feasibility
conditions are specified in Theorem 1 and 2. Figure \ref{fig:Relationship-between-Problem}
summarizes the relationship between Problem 2 and Theorem 1, 2. By
using the necessary conditions in Theorem 1, we first have the following
property for any feasible $\mathcal{L}$ to Problem \ref{Feedback-Dimension-Optimization-2}.
Denote $N_{1}=\textrm{min}(GKd,N)$, $g_{1}=\left\lfloor \frac{G\left((G-1)Kd-M+d\right)}{N_{1}-Kd}\right\rfloor $. 
\begin{lemma}
[Number of Type-II BSs]\label{Maximum-Number-of}Suppose $\mathcal{L}=\left\{ \{m_{jk}:\forall j,k\},g,\{n_{i}:i\in\mathbb{B}_{g}^{I}\}\right\} $
is a feasible solution to Problem \ref{Feedback-Dimension-Optimization-2},
then $\mathcal{L}$ has no more than $(G-g_{1})$ type-II BSs, i.e.,
$g\geq g_{1}$.
\end{lemma}

Lemma \ref{Maximum-Number-of} indicates that we may only allow a
finite number of type-II BSs to satisfy the required IA DoF in the
network.\vspace{-1cm}

\begin{center}
\framebox{\begin{minipage}[t]{1\columnwidth}%
Challenge 3: Design a low-complexity solution to Problem \ref{Feedback-Dimension-Optimization-2}
despite the implicit constraint (\ref{eq:con_1-4}) on $\mathcal{L}$
and the combinatorial nature of the optimization variable ($\mathcal{L}$). %
\end{minipage}}
\par\end{center}

\begin{figure}
\begin{centering}
\includegraphics{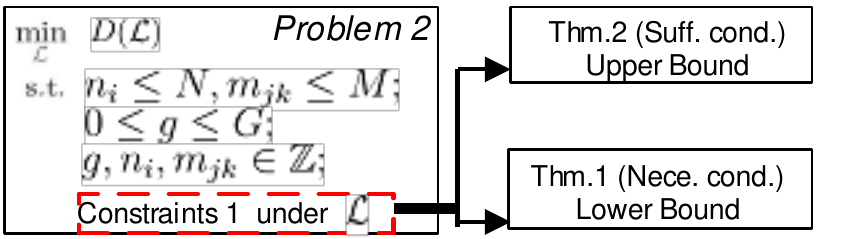}
\par\end{centering}

\caption{\label{fig:Relationship-between-Problem}Relationship between Problem
\ref{Feedback-Dimension-Optimization-2} and Theorem 1, 2.}

\end{figure}

\subsection{Proposed Greedy Algorithm of Feedback Profile Design }

To tackle the challenges, we obtain an achievable upper bound of feedback
dimension by (a) restricting constraint (\ref{eq:con_1-4}) with its
sufficient conditions in Theorem 2 and (b) find a low complexity greedy
algorithm that gives a feedback profile $\mathcal{L}_{0}$ satisfying
the sufficient condition. Specifically, the greedy feedback profile
solution $\mathcal{L}_{0}$ is designed to aggressively select the
\emph{largest} number of \emph{type-II} BSs. While the solution is
a suboptimal upper bound of the minimum feedback dimension $D(\mathcal{L}^{*})$,
we will show later that it is asymptotically optimal as $G\rightarrow\infty$.
Denote $N_{0}=\min\left(GKd,\left\lfloor \frac{N}{d}\right\rfloor d\right)$.
The details of the greedy algorithm are as follows:

\emph{Algorithm 3 (Greedy Solution $\mathcal{L}_{0}$ to Problem \ref{Feedback-Dimension-Optimization-2})}
\begin{itemize}
\item \textbf{Step 1} \emph{(Initialization):} Initialize $\mathcal{L}_{0}=\left\{ \{m_{jk}=M:\forall j,k\},g_{0},\{n_{i}=N_{0}:i\in\mathbb{B}_{g_{0}}^{I}\}\right\} $,
where 
\begin{equation}
g_{0}=\left\lceil \frac{G\left((G-1)Kd-M+d\right)}{N_{0}-Kd}\right\rceil .\label{eq:s_b_expression}
\end{equation}

\item \textbf{Step 2} \emph{(Antenna Pruning Preparation):} Construct the
max flow graph $\mathcal{N}=(\mathcal{V},\mathcal{E})$ \cite{cormen2001introduction}: \end{itemize}
\begin{enumerate}
\item The vertices are given by $\mathcal{V}=\{a,b,u_{jk},v_{i},c_{ji,k}\}$,
$\forall j,k,i\in\mathbb{B}_{g_{0}}^{I}$, where $a$, $b$ are the
source, destination node respectively and $u_{jk},v_{i},c_{ji,k}$
are the intermediate nodes in $\mathcal{N}$. 
\item The edges are given by $\mathcal{E}=\left\{ (a,u_{jk}),(a,v_{i})\right.$,$(u_{jk},c_{jk,i})$,
$(v_{i},c_{jk,i})$, $\left.(c_{jk,i},b):\forall j,k,i\in\mathbb{B}_{g_{0}}^{I}\right\} $,
where $(u,v)$ denotes the edge from node $u$ to node $v$.
\item The edge capacities are given by $c(a,u_{jk})=c(u_{jk},c_{jk,i})=(m_{jk}-\sum_{i\in\mathbb{B}_{g_{0}}^{II}\backslash\{j\}}Kd-d)$,
$c(a,v_{i})=c(v_{i},c_{jk,i})=K(n_{i}-Kd),$ $c(c_{jk,i},t)=Kd$,
$\forall j,k,i\in\mathbb{B}_{g_{0}}^{I}$, where $c(u,v)$ denotes
the edge capacity on the edge $(u,v)$.
\item Find the max flow solutions $\{f(a,b):(a,b)\in\mathcal{E}\}$ for
$\mathcal{N}$ \cite{cormen2001introduction}.\end{enumerate}
\begin{itemize}
\item \textbf{Step 3} \emph{(Antenna Pruning):} Based on the max-flow $\{f(a,b):(a,b)\in\mathcal{E}\}$
obtained in \textbf{Step 2}, perform antenna reduction as 
\[
n_{i}=N_{0}-\left\lfloor \frac{c(a,v_{i})-f(a,v_{i})}{Kd}\right\rfloor d,\, i\in\mathbb{B}_{g_{0}}^{I};
\]
\[
m_{jk}=M-\left\lfloor c(a,u_{jk})-f(a,u_{jk})\right\rfloor ,\forall j,k.
\]
\hfill \QED
\end{itemize}

\begin{remrk}
[Interpretation of Algorithm 3]The feedback profile $\mathcal{L}_{0}$
design in Algorithm 3 contains two stages and in the first stage (Step
1), we design $g_{0}$ in $\mathcal{L}$ by choosing the largest number
of type-II BSs, in the second stage (Step 2, 3), we further reduce
the feedback antennas via max-flow techniques. As the computation
mainly comes from finding the max flow solutions, the overall worst
case complexity of Algorithm 3 is $\mathcal{O}(G^{4}K^{2})$ \cite{cormen2001introduction}. 
\end{remrk}

By deploying Corollary \ref{Equivalent-Condition} and by using the
max-flow graph in Algorithm 3, we derive that $\mathcal{L}_{0}$ satisfies
the conditions in Theorem 2, and is therefore a feasible solution
to Problem \ref{Feedback-Dimension-Optimization-2}.
\begin{thm}
[Feasibility of $\mathcal{L}_{0}$]The obtained feedback profile
$\mathcal{L}_{0}$ from Algorithm 3 is a feasible solution to the
feedback dimension optimization problem (Problem \ref{Feedback-Dimension-Optimization-2}).
\end{thm}

\subsection{Asymptotic Optimality of the Proposed Greedy Solution}

In this section, we further show that $\mathcal{L}_{0}$ is in fact
asymptotically optimal. To do this, we relax the constraint (\ref{eq:con_1-4})
in Problem \ref{Feedback-Dimension-Optimization-2} with its necessary
conditions in Theorem 1, and find a strict lower bound on the minimum
feedback dimension under the necessary conditions (through algebraic
manipulations). Specifically, we have the following bounds on the
optimal feedback dimension.
\begin{thm}
[Bounds on the Optimal Feedback Dimension]\label{Bounds-on-the-dimension}Suppose
$\mathcal{L}^{*}$ is the optimal solution of Problem \ref{Feedback-Dimension-Optimization-2},
then
\begin{equation}
D_{low}\leq D(\mathcal{L}^{*})\leq D(\mathcal{L}_{0})\label{eq:feedback_dimen_bound}
\end{equation}
where $D(\mathcal{L}_{0})$ is the feedback dimension induced by feedback
profile $\mathcal{L}_{0}$ and $D_{low}$ is given by: 
\begin{equation}
D_{low}=KGN_{1}g_{1}\left(M-(G-g_{1})Kd\right)-KG^{2}.\label{eq:lower_bound_fd}
\end{equation}

\end{thm}

From Theorem \ref{Bounds-on-the-dimension}, we derive that $\mathcal{L}_{0}$
can achieve the asymptotic optimality of Problem \ref{Feedback-Dimension-Optimization-2}.
\begin{cor}
[Asymptotic Optimality of $\mathcal{L}_{0}$]\label{Asymptotic-Optimality-of}Suppose
the number of antennas $N$, $M$ are given by $N=\left\lfloor C_{1}KG\right\rfloor $,
$M=\left\lfloor C_{2}KG\right\rfloor $, where $0<C_{1},C_{2}<d$,
$d<C_{1}+C_{2}$. As $G\rightarrow\infty$, we have
\begin{equation}
\lim_{G\rightarrow\infty}\frac{D(\mathcal{L}^{*})}{G^{4}K^{3}}=\lim_{G\rightarrow\infty}\frac{D(\mathcal{L}_{0})}{G^{4}K^{3}}=\frac{\left(d-C_{1}\right)\left(d-C_{2}\right)^{2}}{C_{1}}.\label{eq:asymptotic}
\end{equation}
 \end{cor}
\begin{proof}
As $G\rightarrow\infty$, we have $g_{0}=\frac{d-C_{2}}{C_{1}}G+\mathcal{O}(1)$
and $g_{1}=\frac{d-C_{2}}{C_{1}}G+\mathcal{O}(1)$. Substituting $g_{0}$
and $g_{1}$ into $D(\mathcal{L}_{0})$ and $D_{low}$, we obtain
\[
\lim_{G\rightarrow\infty}\frac{D(\mathcal{L}_{0})}{G^{4}K^{3}}=\lim_{G\rightarrow\infty}\frac{D_{low}}{G^{4}K^{3}}=\frac{\left(d-C_{1}\right)\left(d-C_{2}\right)^{2}}{C_{1}}.
\]
From this and (\ref{eq:feedback_dimen_bound}), the corollary is proved. 
\end{proof}

\begin{remrk}
[Interpretation of Corollary \ref{Asymptotic-Optimality-of}]Corollary
\ref{Asymptotic-Optimality-of} depicts the scaling law of the optimal
feedback dimension w.r.t. the size of the network $G$ and (\ref{eq:asymptotic})
indicates that the proposed greedy solution $\mathcal{L}_{0}$ is
an asymptotically optimal solution to Problem \ref{Feedback-Dimension-Optimization-2}.
Furthermore, using Lemma \ref{Maximum-Number-of} and Corollary \ref{Asymptotic-Optimality-of},
we can infer that the asymptotic optimal $\mathcal{L}_{0}$ has the
largest number of type-II BSs. 
\end{remrk}

From (\ref{eq:asymptotic}), the value of $\lim_{G\rightarrow\infty}\frac{D(\mathcal{L}^{*})}{G^{4}K^{3}}$
gets larger as $d$ increases ($0<C_{1},C_{2}<d$). This agrees with
our intuition that we should pay a larger CSI feedback overhead as
the required IA DoF increases in the network for a given number of
antennas. 
\begin{cor}
[Performance Comparison]Under the same setup as in Corollary 2,
the ratio between the feedback dimension of $\mathcal{L}_{0}$ and
the full CDI feedback scheme (sum feedback dimension $D_{full}=G^{2}K(MN-1)$)
is given by
\begin{equation}
\Upsilon\triangleq\lim_{G\rightarrow\infty}\frac{D(\mathcal{L}_{0})}{D_{full}}=\frac{(d-C_{1})(d-C_{2})^{2}}{(C_{1})^{2}C_{2}}\overset{(x)}{<}1.\label{eq:inequality}
\end{equation}
\hfill \QED
\end{cor}

Note that $(x)$ comes from $\forall i$, $0<C_{1},C_{2}<d$, $d<C_{1}+C_{2}$
as in Corollary \ref{Asymptotic-Optimality-of}. (\ref{eq:inequality})
further implies that larger values of $C_{1}$, $C_{2}$ with $0<C_{1},C_{2}<d$,
$d<C_{1}+C_{2}$ tends to have smaller $\Upsilon$ and hence the proposed
scheme achieves larger CSI feedback reduction gain. This is because
a larger number of antennas at the BSs and MSs (larger $C_{1}$, $C_{2}$)
leads to a larger design space for CSI filtering and hence better
schemes may be obtained.

\section{Relationship Between CSI Feedback Dimension and Feedback Bits}

Recall that in Section II, we propose a novel metric (feedback dimension
$D$) to quantify the effectiveness of CSI feedback filtering. In
this section, we justify the physical meaning of \emph{$D$} in MIMO
cellular networks by deriving the scaling relationship between the
CSI feedback bits $B_{tot}$ and the CSI feedback dimension $D(\mathcal{L})$.
Specifically, we show that when $B_{tot}$ scales with $D$ and SNR
as $B_{tot}=D(\mathcal{L})\log\textrm{SNR}$, the sum DoF of $KGd$
can be achieved in the MIMO cellular network. This result indicates
that the proposed \emph{feedback dimension} can serve as a first-order
measurement of the CSI feedback overhead, and highlights the importance
of feedback dimension optimization in MIMO cellular networks.

\subsection{MIMO Cellular Networks with Limited CSI Feedback Bits}

Suppose that we deploy a feasible feedback profile $\mathcal{L}$
(feasible solution to Problem \ref{Feedback-Dimension-Optimization-2})
in the MIMO cellular network with a total of $B_{tot}$ CSI feedback
bits to quantize and feedback the partial CSI $\{F_{jk}:\forall j,k\}$
generated at the MSs (block (b) in Figure \ref{fig:Role-of-CSI}).
Assume $b$ bits per each feedback dimension and then $B_{tot}=bD(\mathcal{L})$. 

To begin with, we illustrate how the elements in $F_{jk}$ are quantized
using the Grassmannian codebook. We quantize the direction information
$\mathbb{P}(\mathbf{H})$ of the matrix $\mathbf{H}$ by first stacking
$\mathbf{H}$ into a long vector $\textrm{vec}(\mathbf{H})$, and
then quantizing the normalized vector $\mathbf{h}\triangleq\frac{1}{||\mathbf{H}||}\textrm{vec}(\mathbf{H})$
to be $\mathbf{\hat{h}}$, $||\mathbf{\hat{h}}||=1$, with the Grassmannian
vector codebooks \cite{dai2008quantization}. We recover the quantized
version of $\mathbf{H}$ (denoted as $\hat{\mathbf{H}}$) by reverse-stacking
$\mathbf{\hat{h}}$. Based on this quantization approach, we denote
the quantized version of $\mathbf{H}_{jk,i}^{e}$ in $\{F_{jk}\}$
as $\mathbf{\hat{H}}_{jk,i}^{e}$. The relationship between $\mathbf{H}_{jk,i}^{e}$
and $\mathbf{\hat{H}}_{jk,i}^{e}$ can be expressed as 
\[
\mathbf{H}_{jk,i}^{e}=C_{jk,i}\mathbf{\hat{H}}_{jk,i}^{e}+\triangle_{jk,i},\,\forall j,k,i\in\mathbb{B}_{g}^{I}\bigcup\{j\}
\]
where $\{C_{jk,i}\}$ are certain scalars, $\triangle_{jk,i}$ is
the quantization distortion part and $\textrm{vec}(\triangle_{jk,i})$
lies in the orthogonal complement space of $\textrm{vec}(\mathbf{\hat{H}}_{jk,i}^{e})$
\cite{dai2008quantization}. 
\begin{lemma}
[CSI Quantization Distortion]\label{CSI-Quantization-Distortion}Denote
$\mathbb{E}\left(||\triangle_{jk,i}||^{2}\right)$ as the average
quantization distortion of $\mathbf{H}_{jk,i}^{e}$, we have
\[
\mathbb{E}\left(||\triangle_{jk,i}||^{2}\right)=(B_{jk,i}-1)2^{-b},\quad\forall j,k,i\in\mathbb{B}_{g}^{I}\bigcup\{j\}
\]
where $B_{jk,i}$ is given in (\ref{eq:B_definition}).
\end{lemma}

Denote $\{\mathbf{\hat{T}}_{j}^{I}\in\mathbb{U}(N,Kd):j\in\mathbb{B}_{g}^{I}\}$,
$\{\mathbf{\hat{V}}_{jk}^{s}\in\mathbb{U}(Kd,d):\forall j,k\}$, $\{\mathbf{\hat{U}}_{jk}\in\mathbb{U}(N,d):\forall j,k\}$
as the designed outer precoders for type-I BSs, inner precoders for
all BSs, decorrelators for all MSs respectively, based on the quantized
CSI $\left\{ \mathbf{\hat{H}}_{jk,i}^{e}:\forall j,k,i\in\mathbb{B}_{g}^{I}\bigcup\{j\}\right\} $.
Due to the quantization of the feedback CSI, IA cannot be perfectly
achieved and there will be residual interference leakage. Denote the
residual interference covariance matrix at the $(j,k)$-th MS as $\mathbf{\Phi}_{jk}$,
then, 
\begin{equation}
\mathbf{\Phi}_{jk}=\frac{P}{Kd}\sum_{(i,p)\neq(j,k)}\left(\left(\mathbf{\hat{U}}_{jk}^{\dagger}\mathbf{H}_{jk,i}\hat{\mathbf{V}}_{ip}\right)\left(\mathbf{\hat{U}}_{jk}^{\dagger}\mathbf{H}_{jk,i}\hat{\mathbf{V}}_{ip}\right)^{\dagger}\right)\label{eq:interference_covariance_matrix}
\end{equation}
where $\hat{\mathbf{V}}_{ip}=\mathbf{\hat{T}}_{i}^{I}\mathbf{\hat{V}}_{ip}^{s}$,
$i\in\mathbb{B}_{g}^{I}$, $\hat{\mathbf{V}}_{ip}=\mathbf{T}_{i}^{II}\mathbf{\hat{V}}_{ip}^{s}$,
$i\in\mathbb{B}_{g}^{II}$. We have the following lemma on the average
residual interference leakage.
\begin{lemma}
[Residual Interference Bound]\label{Residual-Interference-BoundThe}Denote
$\mathbb{E}\left(\textrm{tr}(\mathbf{\Phi}_{jk})\right)$ as the average
interference leakage, then $\mathbb{E}\left(\textrm{tr}(\mathbf{\Phi}_{jk})\right)$
is upper bounded by
\[
\mathbb{E}\left(\textrm{tr}(\mathbf{\Phi}_{jk})\right)\leq\frac{P}{d}c_{jk}\cdot2^{-b}
\]
where $c_{jk}=\sum_{i\in\mathbb{B}_{g}^{I}\bigcup\{j\}}(B_{jk,i}-1)$.
\end{lemma}

\subsection{Throughput Analysis under Limited CSI Feedback Bits}

Denote $\{\mathbf{T}_{j}\}$, $\{\mathbf{V}_{jk}^{s}:\forall j,k\}$,
$\{\mathbf{U}_{jk}\forall j,k\}$ as the perfect CSIT IA transceivers.
Then the network throughput under perfect CSIT can be expressed as
\cite{gomadam2011distributed},
\begin{equation}
R_{per}=\sum_{j=1}^{G}\sum_{k=1}^{K}\mathbb{E}\left\{ \log\textrm{det}\left(\mathbf{I}_{d}+\frac{P}{Kd}(\mathbf{U}_{jk}^{\dagger}\mathbf{H}_{jk,j}\mathbf{T}_{j}\mathbf{V}_{jk}^{s})(\mathbf{U}_{jk}^{\dagger}\mathbf{H}_{jk,j}\mathbf{T}_{j}\mathbf{V}_{jk}^{s})^{\dagger}\right)\right\} .\label{eq:throughput_per_1-1}
\end{equation}

Following the above definition and treating the residual interference
due to CSI quantization as noise, the network throughput under limited
feedback can be expressed as
\begin{equation}
R_{lim}=\sum_{j=1}^{G}\sum_{k=1}^{K}\mathbb{E}\left\{ \log\textrm{det}\left(\mathbf{I}_{d}+\frac{P}{Kd}(\mathbf{\hat{U}}_{jk}^{\dagger}\mathbf{H}_{jk,j}\hat{\mathbf{V}}_{jk})(\mathbf{\hat{U}}_{jk}^{\dagger}\mathbf{H}_{jk,j}\hat{\mathbf{V}}_{jk})^{\dagger}\left(\mathbf{I}_{d}+\mathbf{\Phi}_{jk}\right)^{-1}\right)\right\} .\label{eq:pratical_throughput_1-1}
\end{equation}
We have the following throughput bounds regarding $R_{lim}$.
\begin{lemma}
[Throughput Bounds]\label{Throughput-Bounds}$R_{lim}$ is bounded
by
\end{lemma}
\begin{equation}
GKd\int_{0}^{\infty}\log\left(1+\frac{P}{Kd}\cdot v\right)\cdot f(v)\textrm{d}v=R_{per}\geq R_{lim}\geq R_{lb}=R_{per}-\sum_{j=1}^{G}\sum_{k=1}^{K}d\cdot\log\left(1+\frac{P}{d^{2}}c_{jk}\cdot2^{-b}\right)\label{eq:throughput-bound}
\end{equation}
where $f(v)$ is the marginal probability density function (p.d.f.)
of the unordered eigenvalues of the \textit{$(d\times d)$ central
Wishart} matrix with $d$ degrees of freedom and covariance matrix
$\mathbf{I}$ ($\mathbf{W}_{d}(\mathbf{I},\; d)$) \cite{tulino2004random}
(pp 32-33).

Based on Lemma \ref{Throughput-Bounds}, we have the following Theorem.
\begin{thm}
[Scaling Law Between CSI Feedback Bits and Feedback Dimension]\label{Feedback-Bits-Saling}When
the total number of CSI feedback bits $B_{tot}$ is given by: 
\begin{equation}
B_{tot}=D(\mathcal{L})\log P\label{eq:bits_scaling}
\end{equation}
the MIMO cellular network can achieve the sum DoF of $GKd$ data streams,
i.e., 
\[
\lim_{P\rightarrow\infty}\frac{R_{lim}}{\log P}=GKd.
\]
\end{thm}
\begin{proof}
From (\ref{eq:bits_scaling}), we obtain $b=\log P$. Hence
\[
R_{lb}=R_{per}-\sum_{j=1}^{G}\sum_{k=1}^{K}d\cdot\log\left(1+\frac{1}{d^{2}}c_{jk}\right).
\]
Note $\sum_{j=1}^{G}\sum_{k=1}^{K}d\cdot\log\left(1+\frac{1}{d^{2}}c_{jk}\right)$
is bounded. Therefore,
\[
\lim_{P\rightarrow\infty}\frac{R_{lb}}{\log P}=\lim_{P\rightarrow\infty}\frac{R_{per}}{\log P}=GKd.
\]
From this and (\ref{eq:throughput-bound}), Theorem \ref{Feedback-Bits-Saling}
is proved. 
\end{proof}

\begin{remrk}
[Interpretation of Theorem \ref{Feedback-Bits-Saling}]Theorem \ref{Feedback-Bits-Saling}
demonstrates a linear scaling relationship between the CSI feedback
bits and feedback dimension in MIMO cellular networks. This result
indicates that the proposed metric of CSI feedback dimension can separate
the CSI filtering and CSI quantization in MIMO cellular networks,
and can serve as a first-order measurement of the feedback overhead.
\end{remrk}

\section{Numerical Results}

In this section, we verify the performance of the proposed feedback
scheme in MIMO cellular networks through simulation. We consider limited
feedback with Grassmannian codebooks \cite{dai2008quantization} to
quantize the partial CSI $\{F_{jk}\}$ at each MS. The precoders /
decorrelators are designed using the Algorithm 1 developed in Section
III-C. We consider $10^{4}$ i.i.d. Rayleigh fading channel realizations
and compare the performance of the proposed feedback scheme with the
following 3 baselines. 
\begin{itemize}
\item \textbf{Baseline 1} \emph{(Feedback Full CDI As in }\cite{thukral2009interference,krishnamachari2009interference,rao2012limited}\emph{):}
Each MS quantizes and feedbacks the full CDI using Grassmannian codebooks,
i.e., $F_{jk}=\left(\begin{array}{ccc}
\cdots, & \mathbb{P}\left(\mathbf{H}_{jk,i}\right), & \cdots\end{array}\right)_{\forall i},\forall j,k$.
\item \textbf{Baseline 2} \emph{(Feedback Truncated CDI As in }\cite{de2012interference}\emph{):}
Each MS quantizes and feedbacks the CDI of the smallest CSI submatrices,
i.e., $F_{jk}=\left(\begin{array}{ccc}
\cdots, & \mathbb{P}(\mathbf{H}_{jk,i}^{s}), & \cdots\end{array}\right)_{\forall i},$ $\forall j,k$, where $\mathbf{H}_{jk,i}^{s}=\left[\begin{array}{cc}
\mathbf{I}_{m} & \mathbf{0}\end{array}\right]\mathbf{H}_{ji,k}$, and $m$ are chosen to make the network tightly feasible by $m=GKd+d-N$
\cite{liu2013feasibility}. 
\item \textbf{Baseline 3} \emph{(Random Beamforming)}: The BS, MS randomly
choose the transceivers $\{\mathbf{T}_{j},\mathbf{V}_{jk}^{s}\}$,
$\{\mathbf{U}_{jk}:\forall j,k\}$. 
\end{itemize}

Consider a MIMO cellular network with $G=3$, $K=2$, $N=M=4$, $d=1$
for simulation tests. We obtain the following feedback profile for
the proposed scheme via Algorithm 3, $\mathcal{L}=\{\{m_{jk}=4:\forall j,k\},g=2,\{n_{1}=4,n_{2}=3\}\}$.
Note the sum feedback dimension for the proposed scheme, baseline
1 and baseline 2 are 114, 198, and 270 respectively under the considered
network topology.

\subsection{Throughput Comparison w.r.t. Transmit SNR}

Figure \ref{fig:Throughput-versus-transmit-snr} illustrates the network
throughput versus the transmit SNR $P$ under a sum feedback bits
of $B_{tot}=800$. The proposed scheme achieves substantial throughput
gain over the baselines. This is because the proposed scheme significantly
reduces the CSI feedback dimension while preserving the IA feasibility.
Under the same number of feedback bits, more CSI feedback bits can
be utilized to reduce the quantization error per dimension. The dramatic
performance gain highlights the importance of reducing the feedback
dimension in MIMO cellular networks. Furthermore, we observe that
the gain is larger at high SNR because residual interference, which
is the major performance bottleneck in high SNR regimes, is significantly
reduced by the proposed scheme. On the other hand, we observe that
the throughputs of all the schemes saturate at high SNR. This is because
under fixed number of CSI feedback bits, the leakage interference
power due to CSI quantization scales with the transmit SNR. 

\begin{figure}
\begin{centering}
\includegraphics[scale=0.6]{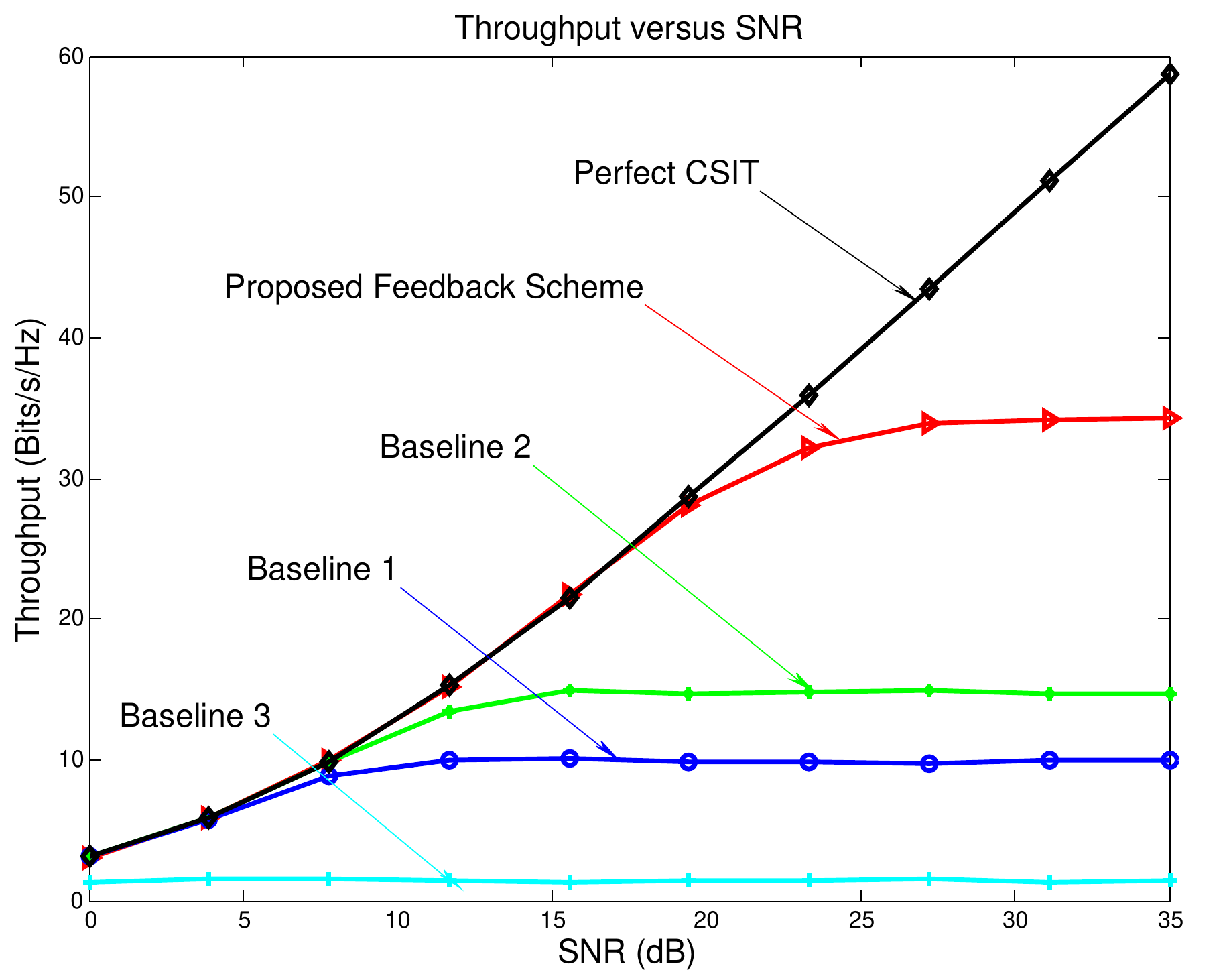}
\par\end{centering}

\caption{\label{fig:Throughput-versus-transmit-snr}Throughput versus transmit
SNR under $B_{tot}=800$ in a $G=3$, $K=2$, $N=M=4$, $d=1$ network.}
\end{figure}

\subsection{Relationship between CSI Feedback Dimension and Feedback Bits}

Figure \ref{fig:Throughput-scaling} illustrates the network throughput
versus the transmit SNR when the number of CSI feedback bits scales
as $B_{tot}=D\log\mbox{SNR}$ as in Theorem \ref{Feedback-Bits-Saling}.
Note $D=114$ as derived for the proposed feedback scheme. As we can
see, the throughput of the proposed scheme achieves the same slope
as that of the perfect CSIT throughput, which justifies that the sum
DoFs of the network are maintained under the given CSI feedback bits
scaling condition as in Theorem \ref{Feedback-Bits-Saling}. However,
the baseline 1, 2 can not achieve the same slope because they have
larger CSI feedback dimension and hence require more feedback bits.

\begin{figure}
\begin{centering}
\includegraphics[scale=0.6]{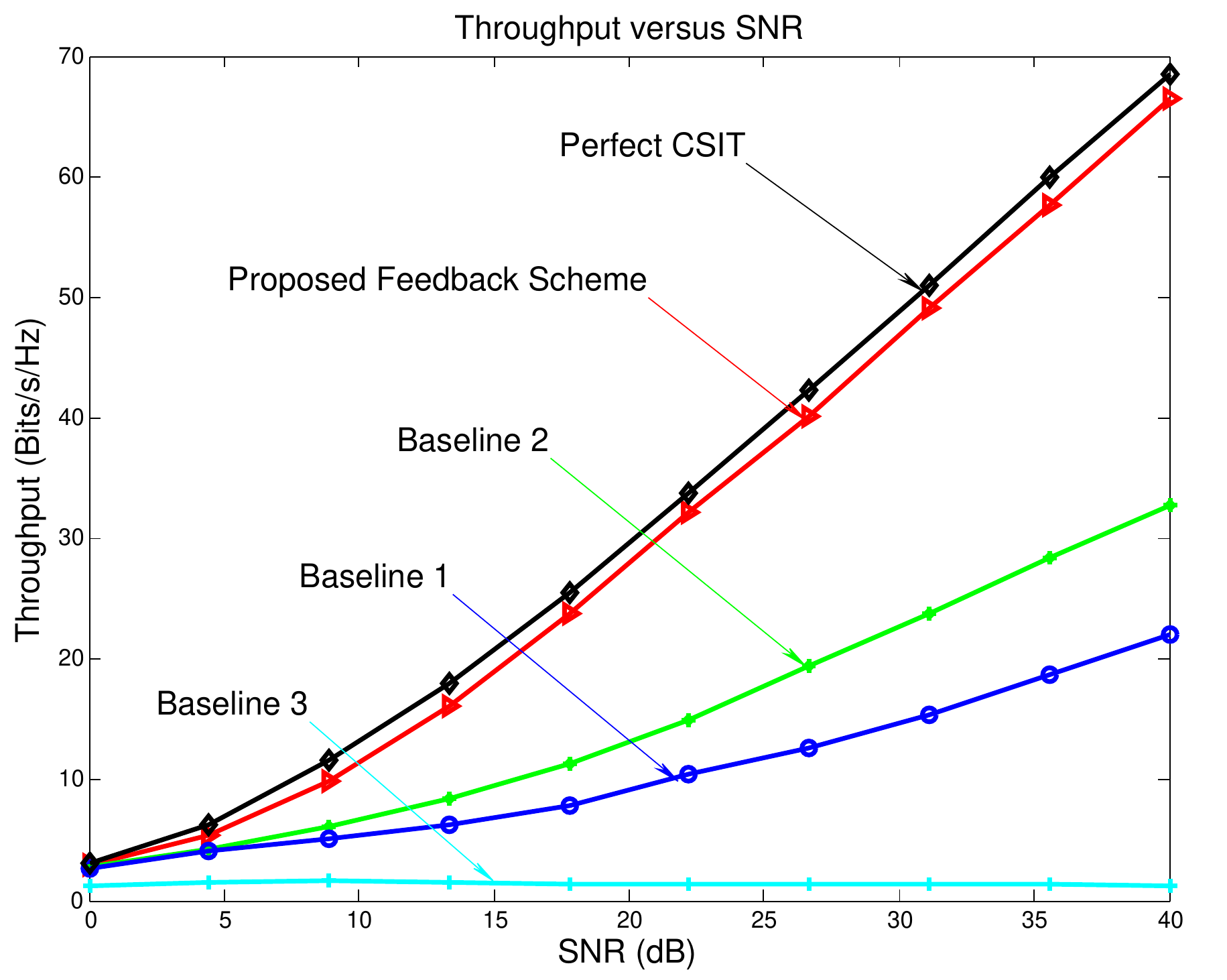}
\par\end{centering}

\caption{\label{fig:Throughput-scaling}Throughput scaling versus transmit
SNR under $B_{tot}=D\log\mbox{SNR}$ in a $G=3$, $K=2$, $N=M=4$,
$d=1$ network. }
\end{figure}

\section{Conclusions}

In this paper, we consider IA processing with CSI feedback filtering
in MIMO cellular networks. We characterize the feedback cost by the
feedback dimension and demonstrate that it can serve as a first order
metric of the CSI feedback overhead. Based on these, we formulate
the problem of feedback dimension minimization subject to the required
IA DoF for a given antenna configuration and we further propose an
asymptotic optimal solution. Both analytical and simulation results
show that the proposed scheme can significantly reduce the CSI feedback
cost of IA in MIMO cellular networks.

\bibliographystyle{ieeetr}
\bibliography{Rx_Zf_IA_ref}

\end{document}